\documentclass[12pt]{article}
\usepackage{amsmath}
\usepackage{graphicx}
\usepackage{enumerate}
\usepackage{natbib}
\usepackage{url} 
\usepackage{siunitx}
\usepackage[utf8]{inputenc} 
\usepackage{amsfonts}       
\usepackage{nicefrac}       
\usepackage{microtype}      
\usepackage{amsmath, amsfonts,bm,amssymb,indentfirst,mathtools,bbm}
\usepackage[noend]{algpseudocode}
\usepackage{algorithmicx,algorithm,enumerate,accents}
\usepackage{color}
\usepackage{upgreek}
\usepackage{titlesec}

\allowdisplaybreaks[4]

\numberwithin{table}{section}
\numberwithin{figure}{section}
\numberwithin{equation}{section}

\newtheorem{theorem}{Theorem}

\newtheorem{assumption}{Assumption}
\newtheorem{corollary}{Corollary}

\newtheorem{remark}{Remark}
\newtheorem{condition}{Condition}
\newtheorem{proposition}{Proposition}

\usepackage{todonotes}

\titleformat{\paragraph}[runin]{\normalfont\normalsize\bfseries}{}{0pt}{}

\newcommand{\blind}{1}

\usepackage{todonotes}

\begin{document}

	\def\spacingset#1{\renewcommand{\baselinestretch}%
		{#1}\small\normalsize} \spacingset{1}

	
	\if1\blind
	{
		\title{\bf Causal Inference under Network Interference with Noise}
		\author{Wenrui Li\thanks{\textit{Contact:} Wenrui Li, wenruili@bu.edu,
				Department of Mathematics and Statistics, Boston University, 111 Cummington Mall, Boston, MA 02215, USA.}\\
			Department of Mathematics \& Statistics, Boston University\\
			Daniel L.
			Sussman \\
			Department of Mathematics \& Statistics, Boston University\\
			and \\
			Eric D.
			Kolaczyk\hspace{.2cm}\\
			Department of Mathematics \& Statistics, Boston University}
		\maketitle
	} \fi
	
	\if0\blind
	{
		\bigskip
		\bigskip
		\bigskip
		\begin{center}
			{\LARGE\bf Causal Inference under Network Interference with Noise}
		\end{center}
		\medskip
	} \fi

	\bigskip
	\begin{abstract}
		Increasingly, there is a marked interest in estimating causal effects under network interference due to the fact that interference manifests naturally in networked experiments.
		However, network information generally is available only up to some level of error.
		We study the propagation of such errors to estimators of average causal effects under network interference.
		Specifically, assuming a four-level exposure model and Bernoulli random assignment of treatment, we characterize the impact of network noise on the bias and variance of standard estimators in homogeneous and inhomogeneous networks.
		In addition, we propose method-of-moments estimators for bias reduction where a minimal number of network replicates are available. We show our estimators are asymptotically normal and provide confidence intervals for quantifying the uncertainty in these estimates. We illustrate the practical performance of our estimators through simulation studies in British secondary school contact networks.
	\end{abstract}
	
	\noindent%
	\vfill
	
	\newpage
	\spacingset{1.5} 
	\section{Introduction}
	\label{sec:intro}

	In recent years, there has been an enormous interest in the assessment of treatment effects within networked systems. Naturally, interference (\cite{cox1958planning}) cannot realistically be assumed away when doing experiments on networks.
	The outcome of one individual may be affected by the treatment assigned to other individuals, which violates the `stable unit treatment value assumption' (SUTVA) (\cite{neyman1923applications}, \cite{rubin1990formal}).
	As a result, much of what is considered standard in the traditional design of randomized experiments and the corresponding analysis for causal inference does not apply directly in this context.
	
	Moreover, network information generally is available only up to some level of error, also known as network noise. For example, there is often measurement error associated with network constructions, where, by `measurement error’ we will mean true edges being observed as non-edges, and vice versa. Such edge noise occurs in self-reported contact networks where participants may not perceive and recall all contacts correctly (\cite{smieszek2012collecting}). It can also be found in biological networks (e.g., of gene regulatory relationships), which are often based on notions of association (e.g., correlation, partial correlation, etc.) among experimental measurements of gene activity levels that are determined by some form of statistical inference.
	We investigate how network noise impacts estimators of average causal effects under network interference and how to account for the noise.
	
	\subsection{Problem setup}
	We assume the observed graph is a noisy version of a true graph.
	Let $G=(V,E)$ be an undirected graph and $G^\text{obs}=(V,E^\text{obs})$ be the observed graph, where we assume that the vertex set $V$ is known. Denote the adjacency matrix of $G$ by $\bm A=(A_{i,j})_{N_v\times N_v}$ and that of $G^\text{obs}$ by $\tilde {\bm A}=(\tilde A_{i,j})_{N_v\times N_v}$.
	Hence $A_{i,j} = 1$ if there is a true edge between the $i$-th vertex and the $j$-th vertex, and 0 otherwise, while $\tilde A_{i,j} = 1$ if an edge is observed between the $i$-th vertex and the $j$-th vertex, and 0 otherwise. We assume throughout that $G$ and $G^\text{obs}$ are simple.
	
	We express the marginal distributions of the $\tilde A_{i,j}$ in the form (\cite{balachandran2017propagation}): 
	\begin{equation}\label{eq3.1}
		\begin{aligned}
			\tilde A_{i,j}\sim
			\begin{cases}
				\text{Bernoulli}(\alpha_{i,j}), & \text{if } \{i,j\}\in E^c,\\
				\text{Bernoulli}(1-\beta_{i,j}), & \text{if } \{i,j\}\in E,\\
			\end{cases}
		\end{aligned}
	\end{equation}
	where $E^c=\{\{i,j\} : i,j\in V; i< j\}  \backslash E$.
	Drawing by analogy on the example of network construction based on hypothesis testing, $\alpha_{i,j}$ can be interpreted as the probability of a Type-I error on the (non)edge status for vertex pair $\{i,j\}\in E^c$, while $\beta_{i,j}$ is interpreted as the probability of Type-II error, for vertex pair $\{i,j\}\in E$.
	Our interest is in characterizing the manner in which the uncertainty in the $\tilde A_{i,j}$ propagates to estimators of average causal effects.

	Let $z_i=1$ indicate that individual $i\in V$ received a given treatment.
	We will refer to $\bm z=(z_1,\cdots,z_{N_v})^\top\in\{0,1\}^{N_v}$ as the treatment assignment vector. Let $p_{\bm z}=\mathbb P(\bm Z=\bm z)$ be the probability that treatment assignment $\bm z$ is generated by the experimental design. Additionally, let $y_i(\bm z)$ denote the outcome for individual $i$ under treatment assignment $\bm z$. In the worst case, there will be $2^{N_v}$ possible exposures for each of the $N_v$ individuals, making causal inference impossible. To avoid this situation, we adopt the notion of so-called exposure mappings, introduced by \citet{aronow2017estimating}.
	We say that $i$ is exposed to condition $k=1,\ldots,K$ if $f(\bm z,\bm x_i ) =c_k$, where $f$ is the exposure mapping, $\bm z$ is the treatment assignment vector, and $\bm x_i $ is a vector of additional information specific to individual $i$.
	Under interference, these authors offer a simple, four-level categorization of exposure ($K=4$) that we revisit here and throughout this paper.
	Taking the vector $\bm x_i$ to be the $i$th column of  the adjacency matrix $\bm  A$ (i.e., $\bm x_i=\bm A_{.i}$), they define
	\begin{equation}\label{eq2.1}
		f(\bm z,\bm A_{\cdot i})=
		\begin{cases}
			c_{11} \text{(Direct + Indirect Exposure)}, & z_i I_{\{\bm z^\top \bm A_{\cdot i}>0\}}=1,\\
			c_{10} \text{(Isolated Direct Exposure)}, & z_i I_{\{\bm z^\top \bm A_{\cdot i}=0\}}=1,\\
			c_{01} \text{(Indirect Exposure)}, & (1-z_i) I_{\{\bm z^\top \bm A_{\cdot i}>0\}}=1,\\
			c_{00} \text{(No Exposure)}, & (1-z_i) I_{\{\bm z^\top \bm A_{\cdot i}=0\}}=1,
		\end{cases}
	\end{equation}
	where the inner product $\bm z^\top \bm A_{\cdot i}$ is the number of treated neighbors of individual $i$.
	
	
	In the general exposure mapping framework of \citet{aronow2017estimating}, potential outcomes are dependent only on the exposure conditions for each unit.
	Suppose each individual $i$ has $K$ potential outcomes $y_i(c_1),\cdots,y_i(c_K)$ and is exposed to one and only one condition.
	Then, define
	\begin{equation}
		\uptau (c_k,c_l)=\frac{1}{N_v}\sum_{i=1}^{N_v}\left[y_i(c_k)-y_i(c_l)\right]=\overline{y}(c_k)-\overline{y}(c_l)
	\end{equation}
	to be the average causal contrast between exposure condition $k$ versus $l$. Consider again, for example, the exposure mapping function defined in (\ref{eq2.1}).
	A natural set of contrasts is $\uptau (c_{01},c_{00})$, $\uptau (c_{10},c_{00})$, and $\uptau (c_{11},c_{00})$, which capture the average indirect treatment effect, the average direct treatment effect, and the average total treatment effect, respectively.

	Now consider the problem of inference for causal effects under network interference. The Horvitz-Thompson framework accounts for unequal-probability sampling through the use of inverse probability weighting (\cite{horvitz1952generalization}) and is adapted by \cite{aronow2017estimating}  under exposure mappings. In noise-free networks, 
	assuming all individuals have nonzero exposure probabilities for all exposure conditions, the estimator
	
	\begin{equation}
		\label{eq:AS.ideal}
		\hat{\overline{y}}(c_k)=\frac{1}{N_v}\Bigg\{\sum_{i=1 }^{N_v}I_{\{f({\bm Z}, {\bm x}_i)=c_k\}} \frac{y_i(c_k)}{p_i^e(c_k)} \Bigg\}
	\end{equation}
	is well-defined and unbiased for $\overline{y}(c_k)$, where the exposure probabilities $p_i^e(c_k)$ are defined as $\sum_{\bm z}   p_{\bm z} I_{\{f({\bm z}, {\bm x}_i)=c_k\}}$.
	In turn, $\hat\uptau(c_k,c_l)=\hat{\overline{y}}(c_k)-\hat{\overline{y}}(c_l)$
	is an unbiased estimator of $\uptau(c_k,c_l)$.
	
	However, in noisy networks, some exposure levels will be misclassified. For example, in the four-level exposure model, for a node $i$, the expected confusion matrix for observed (rows) versus true (columns) exposures has the following form
	
	\begin{equation}\label{eq:confusion}
		\bm P_i\coloneqq\begin{bmatrix} 
			P_i (\tilde c_{11},c_{11})& P_i(\tilde c_{11},c_{10}) & 0&0\\
			P_i(\tilde c_{10},c_{11})  &P_i(\tilde c_{10},c_{10})&0&0\\
			0&0&P_i (\tilde c_{01},c_{01})& P_i (\tilde c_{01},c_{00})  \\
			0&0& P_i(\tilde c_{00},c_{01})  &P_i (\tilde c_{00},c_{00})  \\
		\end{bmatrix},
	\end{equation}
	\begin{sloppypar}
		\noindent where $\tilde c_k$ represents the exposure level in observed networks and $P_i(\tilde {c}_k,c_l)=\mathbb{E}[I_{\{f({\bm Z},\tilde {\bm A}_i)=c_{k}\}} I_{\{f({\bm Z}, {\bm A}_i)=c_{l}\}}]$. The two off-diagonal blocks are equal to 0, since network noise does not affect treatment status. The four symbols $P_i(\tilde {c}_k,c_l), k\neq l$ are cases where exposure levels are misclassified. In the general exposure mapping framework,  the estimators (\ref{eq:AS.ideal}) for $\overline{y}(c_k)$ are in fact 
	\end{sloppypar}
	\begin{equation}\label{equation3.1}
		\tilde{\overline{y}}_{A\&S}(c_k)=\frac{1}{N_v}\sum_{i=1 }^{N_v}I_{\{ \tilde p_i^e(c_k)>0 \} }  I_{\{f({\bm Z}, {\tilde{ \bm X}}_i)=c_k\}}\frac{1}{\tilde p_i^e(c_k)} \Bigg\{\sum_{l=1}^K y_i(c_l) \ I_{\{f({\bm Z}, {\bm{ x}}_i)=c_l\}}\Bigg\},
	\end{equation}
	where $\tilde{\bm{ X}}_i$ is a noisy version of $\bm x_i$, and $\tilde p_i^e(c_k)=\sum_{\bm z}p_{\bm z}I_{\{f({\bm z}, \tilde{\bm X}_i)=c_k\}}$.
	From (\ref{equation3.1}), we can see that the errors introduced into this estimator by network noise come in two forms: incorrect exposure probabilities and misclassified exposure levels.  
	
	In this paper, we will address the following important questions. First, what is the impact of ignoring network noise? Second, how can we account for network noise?

	\subsection{Related literature}
	
	Awareness of interference goes back at least 100 years (e.g., \cite{ross1916application}), and its impact on standard theory and methods has been studied previously in certain specific contexts, including interference localized to an individual across different rounds of treatment in clinical trials with crossover designs (\cite{grizzle1965two}), interference based on spatial proximity of treated units (\cite{kempton1984inter}) and interference within blocks (\cite{hudgens2008toward}).
	For network interference, an assumption that has gained traction is that the causal effects can be passed along edges in the network.
	A highly studied assumption is to assume that unit outcomes are only impacted by their neighbors in the network (\cite{manski2013identification,athey2018exact}).
	Researchers have recently developed frameworks for estimating average unit-level causal effects under network interference.
	For example, \cite{aronow2017estimating} provided unbiased estimators of average unit-level causal effects induced by treatment exposure.
	\cite{sussman2017elements} proposed minimum integrated variance linear unbiased estimators with respect to a distribution on the potential outcomes.

	Extensive work regarding uncertainty analysis has been done in causal inference without the network structure or interference.
	Many studies have explored the effects of uncertainty in propensity scores on causal inference.
	For instance, there have been efforts to develop Bayesian propensity score estimators to incorporate such uncertainties into causal inference (e.g., \cite{an20104}, \cite{alvarez2014uncertain}).
	And there are some studies on the properties for particular matching estimators for average causal effects (e.g., \cite{abadie2006large}, \cite{schafer2008average}).
	But, to our best knowledge, there has been little attention to date given towards uncertainty analysis of estimators for average causal effects under network interference.
	Exceptions include a Bayesian procedure which accounts for network uncertainty and relies on a linear response assumption to increase estimation precision (\cite{toulis2013estimation}), and structure learning techniques to estimate causal effects under data dependence induced by a network represented by a chain graph model, when the structure of this dependence is not known a priori (\cite{bhattacharya2019causal}).
	
	As remarked above, there appears to be little in the way of a formal and general treatment of the error propagation problem in estimators of average causal effects under network interference.
	However, there are several areas in which the probabilistic or statistical treatment of uncertainty enters prominently in network analysis.
	Model-based approaches include statistical methodology for predicting network topology or attributes with models that explicitly include a component for network noise (\cite{jiang2011network}, \citet{jiang2012latent}), the `denoising' of noisy networks (\cite{chatterjee2015matrix}), the adaptation of methods for vertex classification using networks observed with errors (\cite{priebe2015statistical}), a regression model on network-linked data that is
	based on a flexible network effect assumption and is robust to errors in the network structure (\cite{le2020linear}), and a general Bayesian framework for reconstructing networks from observational data (\cite{young2020robust}).
	The other common approach to network noise is based on a `signal plus noise' perspective.
	For example, \cite{balachandran2017propagation} introduced a simple model for noisy networks that, conditional on some true underlying network, assumes we observe a version of that network corrupted by an independent random noise that effectively flips the status of (non)edges.
	Later, \cite{chang2020estimation} developed method-of-moments estimators for the underlying rates of error when replicates of the observed network are available.
	In a somewhat different direction, uncertainty in network construction due to sampling has also been studied in some depth.
	See, for example, \citet[Chapter~5]{kolaczyk2009statistical} or \cite{ahmed2014network} for surveys of this area.
	However, in that setting, the uncertainty arises only from sampling---the subset of vertices and edges obtained through sampling are typically assumed to be observed without error.

	\subsection{Our contributions and organization of the paper}
	Our contribution in this paper is to quantify how network errors propagate to standard estimators of average causal effects under network interference, and to provide new estimators for average causal effects when replicates of the observed network are available.
	Adopting the noise model proposed by \cite{balachandran2017propagation}, we characterize the impact of network noise on the bias and variance of standard estimators (\cite{aronow2017estimating}) under a four-level exposure model and Bernoulli random assignment of treatment, and we illustrate the asymptotic behaviors on networks for varying degree distributions.
	Additionally, we propose method-of-moments estimators of average causal effects that are asymptotically normal (as the number of vertices increases to infinity), when replicates of the observed network are available.
	Numerical simulation in the context of social contact networks in British secondary schools suggests that high accuracy is possible for networks of even modest size.

	The organization of this paper is as follows. In Section~\ref{sec3} we present the bias and variance of standard estimators in noisy networks under a four-level exposure model and Bernoulli random assignment of treatment.
	Section~\ref{sec4} contains our proposed method-of-moments estimators for the true average causal effects.
	Numerical illustrations are reported in Section \ref{sec5}. Finally, we conclude in Section \ref{sec6} with a discussion of future directions for this work. All proofs are relegated to supplementary materials.

	\section{Impact of ignoring network noise}\label{sec3}
	
	In this section, we characterize the impact of network noise on biases and variances of standard estimators under a four-level exposure model and Bernoulli random assignment of treatment.
	Specifically, we show results for two typical classes of networks: homogeneous and inhomogeneous.
	By the term homogeneous we mean the degrees follow a zero-truncated Poisson distribution, and by inhomogeneous, the degrees follow a Pareto distribution with an exponential cutoff \citep{clauset2009power}.
	Note that many real networks present a bounded scale-free behavior with a connectivity cut-off due to the finite size of the network or to the presence of constraints limiting the addition of new links in an otherwise infinite network (\cite{amaral2000classes}).
	The exponential cutoff is most widely used.
	
	\subsection{Network settings and assumptions}

	We consider two typical classes of networks: homogeneous and inhomogeneous. The formal definitions are as follows.
	
	\paragraph{Homogeneous network setting} The degree distribution of $G$ is a zero-truncated Poisson distribution with mean $\bar d$. 
	
	\paragraph{Inhomogeneous network setting} The degree distribution of $G$ is a Pareto distribution with an exponential cutoff with rate $\lambda$, shape $\zeta$, lower bound $d_L$, upper bound $N_v-1$ and mean $\bar d$.
	
	\begin{remark}
		The degree distribution is the probability distribution of the degrees over the whole network.
	\end{remark}
	
	\begin{remark}
		Note that $\bar d$, $\lambda$ and $d_L$ depend on $N_v$. For notational simplicity, we omit $N_v$.
	\end{remark}
	
	\begin{remark}
		In the inhomogeneous network setting,  by the definition of Pareto distribution with an exponential cutoff, the parameters $\lambda$, $\zeta$, $d_L$, $\bar d$ and $N_v$ satisfy the equation
		\begin{align*}
			\bar d =\int_{d_L}^{N_v-1}x\cdot e^{-\lambda x } x^{-(\zeta+1)}dx\Big/\int_{d_L}^{N_v-1} e^{-\lambda x } x^{-(\zeta+1)}dx.
		\end{align*}
	\end{remark}

	Here we focus on a general formulation of the problem in which we make the following assumptions on networks and the treatment assignment. 
	
	\begin{assumption}[Constant marginal error probabilities]\label{a1}
		Assume that \\ $\alpha_{i,j}=\alpha$ and $\beta_{i,j}=\beta$ for all $i< j$, so the marginal error probabilities are $\mathbb P(\tilde A_{i,j}=0|A_{i,j}=1)=\beta$ and $\mathbb P(\tilde A_{i,j}=1|A_{i,j}=0)=\alpha$.
	\end{assumption}
	
	\begin{assumption}[Independent noise]\label{a2}
		The random variables $\tilde A_{i,j}$, for all $i<  j$, are conditionally independent given $A_{i,j}$.
	\end{assumption}
	
	\begin{assumption}[Large Graphs]\label{a3} The number of vertices 
		$N_v\rightarrow\infty$.
	\end{assumption}
	
	In Assumption \ref{a1}, we assume that both $\alpha$ and $\beta$ remain constant over different edges.
	Under Assumptions \ref{a1} and \ref{a2}, the distribution of $\tilde d_i$ is
	\begin{equation*} 
		\begin{aligned}
			\tilde d_i=\sum_{j=1}^{N_v} \tilde A_{j,i} \sim 
			\text{Binomial}(N_v-1-d_i,\alpha ) + \text{Binomial}(d_i,1-\beta ).
		\end{aligned}
	\end{equation*}
	Assumption \ref{a3} reflects both the fact that the study of large graphs is a hallmark of modern applied work in complex networks and, accordingly, our desire to understand asymptotic behaviors of estimators for average causal effects and provide concise descriptions in terms of biases and variances for large graphs.

	\begin{assumption}\label{a4}
		Individuals are assigned treatment independently with probability $p$ satisfying $p=o(1)$, $p=\omega(1/N_v)$, $\bar d=\Theta(1/p)$.
		Letting $C_{ij}$ denote the number of common neighbors between vertices $i$ and $j$ in $G$, $\sum_{i=1}^{N_v}\sum_{j\neq i}^{N_v} I_{\{C_{ij}=0\}} \sim N_v^2$.
		Finally, the potential outcomes are bounded, $|y_i(c_k)|\leq c<\infty$, for all values $i$ and $c_k$, where $c$ is a constant.
	\end{assumption}
	
	Assumption \ref{a4} entails that, as $N_v$ grows, the expected number of treated individuals also grows but is dominated by $N_v$ asymptotically.
	And the average number of treated neighbors is bounded.
	The amount of vertex pairs having common neighbors is also limited in scope as $N_v$ grows which ensures a sufficiently large set of independent exposures. Assumption \ref{a4} is an assumption used in proving the consistency of $\hat{\uptau}(c_k,c_l)$ in noise-free homogeneous and inhomogeneous networks. See Appendix \ref{app1} for details. 
	
	\begin{assumption}\label{a5}
		$1-\beta=\Omega(1)$, $\alpha=\Theta(1/(pN_v))$, and $\alpha=o(p)$.
		
	\end{assumption}
	\begin{remark}
		Note that $\alpha$ and $\beta$ can be constants or $o(1)$ as $N_v\rightarrow \infty$.
		For notational simplicity, we omit $N_v$.
	\end{remark}
	\begin{remark}
		Assumption \ref{a5} implies $p=\omega(1/\sqrt{N_v})$, which is consistent with Assumption \ref{a4}.
	\end{remark}
	
	By making assumptions on the underlying rates of error $\alpha$ and $\beta$, we will see that regularity conditions hold for noisy homogeneous and inhomogeneous networks in Appendix \ref{app2} .

	\subsection{Biases of standard estimators in noisy networks}\label{sec3.1}
	
	Assuming a four-level exposure model and Bernoulli random assignment of treatment, we quantify the biases of standard estimators in homogeneous and inhomogeneous network settings.  We begin with the following general result.
	
	\begin{theorem} \label{th1}
		Assume a four-level exposure model and Bernoulli random assignment of treatment with probability $p$.
		Under Assumptions \ref{a1} -- \ref{a3}, \ref{a5}, $p=o(1)$, $p=\omega(1/N_v)$ and the potential outcomes are bounded, we have
		\begin{align*}
			\text{Bias}\Big[ \tilde{\overline{y}}_{A\&S}(c_{11})\Big] =& \ -  \frac{1}{N_v} \sum_{i=1}^{N_v} \frac{(1-p)^{d_i}\big[1-(1-\alpha p)^{N_v-1-d_i}\big]}{1-(1-\alpha p)^{N_v-1-d_i}(1- (1-\beta)p  )^{d_i}}\ \uptau_i(c_{11},c_{10})+  o(1),\\
			\text{Bias}\Big[ \tilde{\overline{y}}_{A\&S}(c_{10})\Big]=  &\ \frac{1}{N_v} \sum_{i=1}^{N_v}\big[1-(1-\beta p)^{d_i}\big] \ \uptau_i(c_{11},c_{10}),\\
			\text{Bias}\Big[ \tilde{\overline{y}}_{A\&S}(c_{01})\Big] =& \ -  \frac{1}{N_v} \sum_{i=1}^{N_v} \frac{(1-p)^{d_i}\big[1-(1-\alpha p)^{N_v-1-d_i}\big]}{1-(1-\alpha p)^{N_v-1-d_i}(1- (1-\beta)p  )^{d_i}}\ \uptau_i(c_{01},c_{00}) +  o(1),\\
			\text{Bias}\Big[ \tilde{\overline{y}}_{A\&S}(c_{00})\Big]=  &\ \frac{1}{N_v} \sum_{i=1}^{N_v}\big[1-(1-\beta p)^{d_i}\big] \ \uptau_i(c_{01},c_{00}),
		\end{align*}
		as $N_v\rightarrow \infty$, where $\uptau_i(c_k,c_l)=y_i(c_k)-y_i(c_l)$ and $d_i$ is the degree of the $i$-th vertex in the noise-free network $G$.
	\end{theorem}
	
	Theorem \ref{th1} then directly leads to the following corollary in homogeneous and inhomogeneous network settings.
	
	\begin{corollary}[Homogeneous and inhomogeneous] \label{coro1}
		Assume a four-level exposure model and Bernoulli random assignment of treatment with $p$.
		In both homogeneous and inhomogeneous network settings, under Assumptions \ref{a1} -- \ref{a3}, \ref{a5}, $p=o(1)$, $p=\omega(1/N_v)$ and the potential outcomes are bounded, the bias statement in Theorem \ref{th1} holds.
	\end{corollary}
	The proof of Theorem \ref{th1} is in supplementary material C.
	Corollary \ref{coro1} directly follows from Theorem \ref{th1}.
	
	The above results show that biases of standard estimators in homogeneous and inhomogeneous network settings have the same expressions.
	Biases of $\tilde{\overline{y}}_{A\&S}(c_{11})$ and $\tilde{\overline{y}}_{A\&S}(c_{01})$ depend on both $\alpha$ and $\beta$, while biases of $\tilde{\overline{y}}_{A\&S}(c_{10})$ and $\tilde{\overline{y}}_{A\&S}(c_{00})$ only depend on $\beta$.
	Biases of $\tilde{\overline{y}}_{A\&S}(c_{11})$ and $\tilde{\overline{y}}_{A\&S}(c_{10})$ are related to $\uptau(c_{11},c_{10})$.
	And biases of $\tilde{\overline{y}}_{A\&S}(c_{01})$ and $\tilde{\overline{y}}_{A\&S}(c_{00})$ are related to $\uptau(c_{01},c_{00})$.
	These relationships follow because the network noise affects observed edges but not treatment status.
	
	\label{para}
	In proving these results, we also necessarily obtain an understanding of estimation bias for causal effects at the level of individuals, which we summarize here.  Let $\tilde y_{A\&S,i}(c_k)$ denote the Aronow and Samii estimator for $y_i(c_k)$ in noisy networks, which corresponds to the $i$-th element of $\tilde{\overline{y}}_{A\&S}(c_k)$ in (\ref{equation3.1}). We summarize in the following table the asymptotic biases of $\tilde y_{A\&S,i}(c_k)$ for high (top row) and low (bottom row) degree nodes.
	
	\begin{table}[!h]
		\caption{The asymptotic biases of $\tilde y_{A\&S,i}(c_k)$ for high (top row) and low (bottom row) degree nodes.}
		\begin{center}
			\begin{tabular}{l|llll}
				\hline  
				& $\text{Bias}[\tilde y_{A\&S,i}(c_{ 11})]$ &  $\text{Bias}[\tilde y_{A\&S,i}(c_{ 10})]$&  $\text{Bias}[\tilde y_{A\&S,i}(c_{ 01})]$ &  $\text{Bias}[\tilde y_{A\&S,i}(c_{ 00})]$\\\hline
				$d_i=\omega(1/p)$ & $ o(1)$ & $\uptau_i(c_{11},c_{10})$ & $ o(1)$& $\uptau_i(c_{01},c_{00})$ \\
				$d_i=o(1/p)$ & $-\uptau_i(c_{11},c_{10})$  & $ o(1)$   & $-\uptau_i(c_{01},c_{00})$ & $ o(1)$ \\
				\hline
			\end{tabular}
		\end{center}
		\label{tab1}
	\end{table}
	We see that there are four cases where $\tilde y_{A\&S,i}(c_k)$ is asymptotically unbiased. The reason is that the corresponding entries in the expected confusion matrix (\ref{eq:confusion}) go to 0. For the other four cases, the corresponding entries in the expected confusion matrix approach 1, which leads to nontrivial biases. Note that the asymptotic biases of $\tilde y_{A\&S,i}(c_k)$ is between 0 and the corresponding $\pm\uptau_i(c_k,c_l)$ when $d_i=\Theta(1/p)$.
	
	\subsection{Variances of standard estimators in noisy networks}\label{sec3.2}
	
	We analyze the variances of standard estimators in homogeneous and inhomogeneous network settings.
	
	\begin{theorem}[Homogeneous]\label{th3}
		Assume a four-level exposure model and Bernoulli random assignment of treatment with probability $p$.
		In the homogeneous network setting, under Assumptions \ref{a1} - \ref{a5}, for all $c_k$, we have $\text{Var}[ \tilde{\overline{y}}_{A\&S}(c_{k})]=o(1)$ as $N_v\rightarrow\infty$.
	\end{theorem}
	
	\begin{theorem}[Inhomogeneous]\label{th4}
		Assume a four-level exposure model and Bernoulli random assignment of treatment with probability $p$.
		In the inhomogeneous network setting, under Assumptions \ref{a1} - \ref{a5}, $\lambda=\Theta(p)$ and $\lambda>p$, we have $\text{Var}[ \tilde{\overline{y}}_{A\&S}(c_{k})]=o(1)$ for all $c_k$ as $N_v\rightarrow\infty$.
	\end{theorem}
	
	Note that the variances go to zero as the number of nodes tends towards infinity for both cases.
	Therefore, in noisy networks, the bias would appear to be the primary concern for estimating average causal effects.

	\section{Accounting for network noise}\label{sec4}
	
	As we saw in Section \ref{sec3}, standard estimators are biased in both homogeneous and inhomogeneous network settings.
	Thus, it is important to have new estimators for bias reduction.
	We present method-of-moments estimators in Section \ref{sec4.1}, and show unbiasedness and consistency under a four-level exposure model and Bernoulli random assignment of treatment in Section \ref{sec4.2}.
	The method-of-moments estimators require either knowledge of or consistent estimators of $\alpha$ and $\beta$.
	For our numerical work in Section~\ref{sec5}, we adopt the estimators in \cite{chang2020estimation}, which require at least three replicates of the observed network.

	\subsection{Method-of-moments estimators}\label{sec4.1}
	
	We construct method-of-moments estimators (MME) by reweighting the observed outcomes based on the expected confusion matrix.
	For convenience, we denote
	\begin{align*}
		{\bm{y}}_i&= [  y_i(c_{11}),   y_i(c_{10}), y_i(c_{01}) , y_i(c_{00})]^\top,\\
		\mathbbm{1}({\bm x}_i) &=[I_{\{f({\bm Z}, {\bm x}_i)=c_{11}\}},I_{\{f({\bm Z}, {\bm x}_i)=c_{10}\}},I_{\{f({\bm Z}, {\bm x}_i)=c_{01}\}},I_{\{f({\bm Z}, {\bm x}_i)=c_{00}\}}]^\top,
		\\
		\mathbbm{1}(\tilde {\bm X}_i) &=[I_{\{f({\bm Z},\tilde {\bm X}_i)=c_{11}\}},I_{\{f({\bm Z},\tilde {\bm X}_i)=c_{10}\}},I_{\{f({\bm Z},\tilde {\bm X}_i)=c_{01}\}},I_{\{f({\bm Z},\tilde {\bm X}_i)=c_{00}\}} ]^\top.
	\end{align*}
	We then combine the observed outcome $\mathbbm{1}({\bm x}_i)^\top  {\bm{y}}_i $  and the observed exposure level into a vector, denoted by $\tilde {\bm{y}}_i$,
	\begin{equation}
		\tilde {\bm{y}}_i= \mathbbm{1}(\tilde {\bm X}_i) \cdot \mathbbm{1}({\bm x}_i)^\top {\bm{y}}_i.
	\end{equation}
	By taking the expectation with respect to treatment and network noise, we obtain
	\begin{equation}
		\mathbb{E}[\tilde {\bm{y}}_i]= \bm P_i\cdot {\bm{y}}_i,
	\end{equation}
	Note that $\bm P_i$ depends on $d_i$, $\alpha$ and $\beta$. Therefore, we use $\bm P (d_i,\alpha,\beta)$ for explicitness.

	Our method of moments estimator for ${\bm{y}}_i$ is defined as 
	\begin{equation}  \label{eq4.4}
		\tilde {{\bm{y}}}_{\text{MME},i}={\bm P} ^{-1} ( \hat d_i, \hat \alpha,\hat \beta)\cdot \tilde {\bm{y}}_i,
	\end{equation}
	where
	\begin{equation} \label{eq4.5}
		\hat d_i=\frac{\tilde d_i-(N_v-1)\hat\alpha}{1-\hat\alpha-\hat\beta}.
	\end{equation}
	The values $\hat\alpha$ and $\hat\beta$ are assumed to be consistent estimators of $\alpha$ and $\beta$, examples of which we provide later. If $\alpha$ and $\beta$ are known, we substitute those values for $\hat\alpha$ and $\hat\beta$ in (\ref{eq4.4}), and this does not change the asymptotic behavior we state in Section \ref{sec4.2}.
	
	We define the method-of-moments estimator for the average potential outcome $\sum_{i=1}^{N_v} {y}_i(c_k)/N_v$  
	\begin{align}\label{eq:mme}
		\begin{split}
			\tilde {\overline{ {y}}}_\text{MME}(c_k)= &\  \frac{1}{N_v}\sum_{i=1}^{N_v}\Bigg\{ \tilde {{{y}}}_{ \text{MME},i}(c_k)\cdot I_{\{\hat d_i=\Theta(1/p)\}}+\tilde {{ {y}}}_{ \text{A\&S},i}(c_k)\cdot I_{\{\hat d_i= \omega(1/p)\bigcap c_k\in\{c_{11},c_{01}\}  \}}\\
			&+\tilde {{{y}}}_{ \text{A\&S},i}(c_k)\cdot I_{\{\hat d_i=o(1/p)\bigcap c_k\in\{c_{10},c_{00}\}   \}}\Bigg\},
		\end{split}
	\end{align}
	where $\tilde {{ {y}}}_{ \text{A\&S},i}(c_k)$ is the Aronow and Samii estimator of node $i$ in the noisy network. Recall from the bias statements in Table~\ref{tab1} that $\tilde y_{A\&S,i}(c_{ 11})$ and $\tilde y_{A\&S,i}(c_{01})$ are asymptotically unbiased for nodes with high degrees.
	And $\tilde y_{A\&S,i}(c_{ 10})$ and $\tilde y_{A\&S,i}(c_{00})$ are asymptotically unbiased for small degree nodes.
	Therefore, we do not need to correct biases for those cases.
	We will show that $\tilde {{\bm{y}}}_{ \text{MME},i}$ is asymptotically unbiased with small variance for nodes with degree on the order of $1/p$ in Theorems \ref{th5} and \ref{th6}.
	Otherwise, asymptotically unbiased estimators with small variances may not exist due to the structure of this specific four-level exposure model.
	As we saw, $\mathbb E[\tilde y_{A\&S,i}(c_{ 11})]\rightarrow y_i(c_{10})$ and $\mathbb E[\tilde y_{A\&S,i}(c_{ 01})]\rightarrow y_i(c_{00})$ for small degree nodes, while $\mathbb E[\tilde y_{A\&S,i}(c_{ 10})]\rightarrow y_i(c_{11})$ and $\mathbb E[\tilde y_{A\&S,i}(c_{ 00})]\rightarrow y_i(c_{01})$ for high degree nodes.
	These means that we lose almost all information about $y_i(c_{11})$ and $y_i(c_{01})$ for small degree nodes, and $y_i(c_{10})$ and $y_i(c_{00})$ for high degree nodes.
	
	In general, we suggest to use terms of the same orders of magnitude in (\ref{eq:mme}) to approximate $\Theta(\cdot)$.  That is, writing 
	$1/p=a\times 10^b$, where $1/\sqrt{10} \leq a<\sqrt{10}$, we represent the order of magnitude with $b$.
	Next, we rewrite $\tilde {\overline{{y}}}_\text{MME}(c_k)$ as
	\begin{align*}
		\begin{split}   
			\tilde {\overline{{y}}}_\text{MME}(c_k)=&\ \frac{1}{N_v}\sum_{i=1}^{N_v}\Bigg\{ \tilde {{ {y}}}_{ \text{MME},i}(c_k)\cdot I_{\{C_1\leq\hat d_i<C_2\}}+\tilde {{ {y}}}_{ \text{A\&S},i}(c_k)\cdot I_{\{\hat d_i \geq C_2 \bigcap c_k\in\{c_{11},c_{01}\} \}}\\
			&+\tilde {{ {y}}}_{ \text{A\&S},i}(c_k)\cdot I_{\{\hat d_i<C_1 \bigcap c_k\in\{c_{10},c_{00}\} \}}\Bigg\},
		\end{split}
	\end{align*}
	where  $C_1=10^b/\sqrt{10}$  and $C_2=\sqrt{10}\cdot 10^b$.
	For sparse networks with small sample sizes, $C_1$ may be close to the average degree and thus we recommend to compute
	\begin{equation}\label{eq:mme2}
		\tilde {\overline{{y}}}_\text{MME}(c_k)= \frac{1}{N_v}\sum_{i=1}^{N_v}\Bigg\{ \tilde {{{y}}}_{ \text{MME},i}(c_k)\cdot I_{\{\hat d_i\geq 1\}}+\tilde {{{y}}}_{ \text{A\&S},i}\cdot I_{\{\hat d_i<1\}}(c_k)\Bigg\}.
	\end{equation}
	
	\begin{remark}
		As we will see later, in this specific four-level exposure model, $\tilde {\overline{ {y}}}_\text{MME}(c_k)$ is asymptotically unbiased and consistent in both homogeneous and inhomogeneous network settings.
		
	\end{remark}

	Our estimators require knowledge of or, more realistically, consistent estimates of the parameters $\alpha$ and $\beta$ governing the noise.  For our numerical work in Section~\ref{sec5}, we adopt the consistent MME estimators in \cite{chang2020estimation}, which require at least three replicates of the observed network.
	Define relevant quantities as follows:
	\begin{align*}
		u_1&= (1-\delta)\alpha+\delta(1-\beta),\\
		u_2&= (1-\delta)\alpha(1-\alpha)+\delta\beta(1-\beta),\\
		u_3&= (1-\delta)\alpha(1-\alpha)^2+\delta\beta^2(1-\beta),
	\end{align*}
	where $\delta$ is the edge density in the true network $G$, $u_1$ is the expected edge density in one observed network, $u_2$ is the expected density of edge differences in two observed networks, and $u_3$ is the average probability of having an edge between two arbitrary nodes in one observed network but no edge between the same nodes in the other two observed networks.
	The method-of-moments estimators for $u_1$, $u_2$ and $u_3$ are
	\begin{equation}\label{eq5.3}
		\begin{aligned}
			\hat u_1&=\frac{2}{N_v(N_v-1)}\sum_{i<j}\tilde A_{i,j}, \\
			\hat u_2&=\frac{1}{N_v(N_v-1)}\sum_{i<j}|\tilde A_{i,j,*}-\tilde A_{i,j}|,\\
			\hat u_3&=\frac{2}{3N_v(N_v-1)}\sum_{i<j} I( \text{Exactly one of }\tilde A_{i,j,**}, \tilde A_{i,j,*}, \tilde A_{i,j} \text{ equals } 1),
		\end{aligned}
	\end{equation} 
	where $ \tilde {\bm A}_*=(\tilde A_{i,j,*})_{N_v\times N_v},\ \tilde {\bm A}_{**}=(\tilde A_{i,j,**})_{N_v\times N_v}$ are independent and identically distributed replicates of $\tilde {\bm A}$.
	Calculation of the estimators $\hat\alpha$ and $\hat\beta$ can be accomplished as detailed in Algorithm \ref{algo1} below.
	
	\begin{algorithm}[!h] 
		\caption{Consistent estimators $\hat\alpha$ and $\hat\beta$} 
		\hspace*{0.02in} {\bf Input:} 
		$\tilde {\bm A}=(\tilde A_{i,j})_{N_v\times N_v}, \ \tilde {\bm A}_*=(\tilde A_{i,j,*})_{N_v\times N_v},\ \tilde {\bm A}_{**}=(\tilde A_{i,j,**})_{N_v\times N_v},\ \alpha_0,\ \varepsilon$\\
		\hspace*{0.02in} {\bf Output:} 
		$\hat\alpha$, $\hat\beta$ 
		\begin{algorithmic}  
			\State Compute $\hat u_1,\ \hat u_2,\  \hat u_3$ defined in (\ref{eq5.3});
			\State Initialize $\hat \alpha=\alpha_0$, $\alpha_0=\hat \alpha+10\varepsilon$;
			\While {$|\hat \alpha-\alpha_0|>\varepsilon$}
			\State $\alpha_0\gets\hat \alpha,\ \hat \beta\gets\frac{\hat u_2-\alpha_0+\hat u_1 \alpha_0}{\hat u_1-\alpha_0},\ \hat\delta\gets\frac{(\hat u_1-\alpha_0)^2}{\hat u_1-\hat u_2-2\hat u_1\alpha_0+\alpha_0^2},\ \hat\alpha\gets\frac{\hat u_3-\hat\delta\hat\beta^2(1-\hat\beta)}{(1-\hat\delta)(1-\alpha_0)^2}$.
			\EndWhile
		\end{algorithmic}
		\label{algo1}
	\end{algorithm}
	
	\subsection{Asymptotic unbiasedness, consistency and normality}\label{sec4.2}
	
	We consider the asymptotic behavior of the  method-of-moments estimators $\tilde {\overline{y}}_\text{MME}(c_k)$ as $N_v\rightarrow\infty$.
	
	\begin{theorem}[Homogeneous]\label{th5}
		Assume a four-level exposure model and Bernoulli random assignment of treatment with probability $p$.
		In the homogeneous network setting, under Assumptions \ref{a1} - \ref{a5}, $\tilde {\overline{y}}_\text{MME}(c_k)$ is an asymptotically unbiased and consistent estimator of $\overline{y}(c_k)$ for all $c_k$.
	\end{theorem}
	
	\begin{theorem}[Inhomogeneous]\label{th6}
		Assume a four-level exposure model and Bernoulli random assignment of treatment with $p$.
		In the inhomogeneous network setting, under Assumptions \ref{a1} - \ref{a5}, $\lambda=\Theta(p)$ and $\lambda>p$, $\tilde {\overline{y}}_\text{MME}(c_k)$ is an asymptotically unbiased and consistent estimator of $\overline{y}(c_k)$ for all $c_k$.
	\end{theorem}
	
	Note that $\tilde {\overline{y}}_\text{MME}(c_k)$ is an asymptotically unbiased and consistent estimator of $\overline{y}(c_k)$ in both homogeneous and inhomogeneous network settings.
	Proofs of Theorem \ref{th5} and \ref{th6} appear in supplementary material C.
	
	Next, we establish assumptions for the asymptotic normality. Let $C_{\mathcal V_1}$ and $C_{\mathcal V_2}$ denote the two-stars count and the count of 3 connected edges passing through 4 different nodes in $G$, respectively. Then, 
	\begin{align*} 
		C_{\mathcal V_1} = \sum_{\bm v= (i_1,i'_1,i_2,i'_2)\in\mathcal V_1 }  A_{ i_1,i'_1 }A_{ i_2,i'_2 } 
	\end{align*}
	and 
	\begin{align*} 
		C_{\mathcal V_2} = \sum_{\bm v= (i_1,i'_1,i_2,i'_2,i_3,i'_3)\in\mathcal V_2 }  A_{ i_1,i'_1 }A_{ i_2,i'_2 } A_{ i_3,i'_3 } 
	\end{align*}
	where $\mathcal V_1= \{ (i_1,i'_1,i_2,i'_2): i'_1=i_2,i_1\neq i_2\neq i'_2\} $ and $\mathcal V_2= \{ (i_1,i'_1,i_2,i'_2,i_3,i'_3): i'_1=i_2, i'_2=i_3,i_1\neq i_2 \neq i_3\neq i'_3\} $.
	
	\begin{assumption}\label{a6}
		$C_{\mathcal V_1}=\mathcal O(N_v^{5/3})$ and $C_{\mathcal V_2}=\mathcal O(N_v^{20/9})$.  
	\end{assumption}
	
	\begin{remark}
		Note that $C_{\mathcal V_1}=\Theta(N_v(\bar d )^2)$ in both homogeneous and inhomogeneous network settings. Thus, Assumption \ref{a6} implies $\bar d=\mathcal O(N_v^{1/3})$, which is consistent with Assumptions \ref{a4} and \ref{a5}.
	\end{remark}
	
	Assumption \ref{a6} is a condition on the connectivity of the true underlying network, induced by an assumption of local dependency of the observations $y_i(c_k)$. This assumption is actually a relaxation of that assumed by \cite{aronow2017estimating}. Their local dependence condition implies bounded network degrees in the four-level exposure model and Bernoulli random assignment of treatment setting. Thus, it leads to $C_{\mathcal V_1}=\mathcal O (N_v)$ and $C_{\mathcal V_2}=\mathcal O(N_v)$. Our assumption allows network degrees to grow as $N_v\rightarrow\infty$, and relaxes the upper bounds on $C_{\mathcal V_1}$ and $C_{\mathcal V_2}$.
	
	\begin{assumption}\label{a8}
		$\text{Var}(\tilde{\bar y}_\text{MME}(c_k ))=\omega(1/(\bar{d})^{2/3})$.  
	\end{assumption}
	
	Assumption \ref{a8} entails a moderate variance of our estimator, ruling out the cases where the bias is relatively large.
	
	\begin{theorem}[Homogeneous]\label{th8}
		Assume a four-level exposure model and Bernoulli random assignment of treatment with $p$.
		In the homogeneous network setting, under the conditions of Theorem \ref{th5}, as well as Assumptions  \ref{a6} and \ref{a8}, we have $\left(\tilde {\overline{y}}_\text{MME}(c_k)-\mathbb E(\tilde {\overline{y}}_\text{MME}(c_k)) \right)/\sqrt{\text{Var}(\tilde{\bar y}_\text{MME}(c_k ))}\xrightarrow{d}N(0,1)$ for all $c_k$.
	\end{theorem}
	
	\begin{theorem}[Inhomogeneous]\label{th9}
		Assume a four-level exposure model and Bernoulli random assignment of treatment with $p$.
		In the inhomogeneous network setting, under the conditions of Theorem \ref{th6}, as well as Assumptions  \ref{a6} and \ref{a8}, we have $\left(\tilde {\overline{y}}_\text{MME}(c_k)-\mathbb E(\tilde {\overline{y}}_\text{MME}(c_k))\right)/\sqrt{\text{Var}(\tilde{\bar y}_\text{MME}(c_k ))}\xrightarrow{d}N(0,1)$ for all $c_k$.
	\end{theorem}
	
	Proofs of Theorems \ref{th8} and \ref{th9} may be found in supplementary material C.  We note that while $\tilde {\overline{y}}_\text{MME}(c_k)$ is asymptotically normal in both homogeneous and inhomogeneous network settings, the bias of $\tilde {\overline{y}}_\text{MME}(c_k)$ is the driver in the inhomogeneous network setting.  This is in contrast to the homogeneous network setting, for which we have the following corollary (also proved in supplementary material C).
	\begin{corollary}[Homogeneous]\label{coro2}
		Under the conditions of Theorem~\ref{th8}, the same asymptotic behavior holds with $\mathbb E(\tilde {\overline{y}}_\text{MME}(c_k))$ replaced by ${\overline{y}}(c_k)$.
	\end{corollary}
	
	
	\subsection{Bias and Variance estimation}\label{sec3.3}
	
	Due to the unknown structure of the true underlying network and the form of our method-of-moments estimators, it is hard to get a closed-form unbiased bias estimator and an unbiased or good conservative and variance estimator. A causal bootstrap has been developed in the context of the potential outcomes framework and under SUTVA by \cite{10.1214/20-AOS2009}, for the purpose of approximating the properties of average treatment effect estimators. A generic bootstrap method has been proposed in the context of contact networks by \cite{kucharski2018structure}, for the goal of assessing various summaries of network structure. Inspired by these two bootstrap methods, we propose bootstrap estimators for $\text{Var}(\tilde{\bar y}_\text{MME}(c_k ))$ and $\text{Bias}(\tilde{\bar y}_\text{MME}(c_k ))$. when a minimum of two replicates of the observed network are available.
	
	Suppose we have $m$ replicates $\tilde {\bm A}^{(1)}=(\tilde A^{(1)}_{i,j})_{N_v\times N_v},\cdots,\tilde {\bm A}^{(m)}=(\tilde A^{(m)}_{i,j})_{N_v\times N_v}$, where $m\geq 2$. Let $\tilde Y_i$ be the observed outcome for individual $i$. Let $\tilde y_i(c_k)$ be the potential outcome for the exposure level $c_k$ in the observed network for individual $i$. Note that $\tilde y_i(c_k)$ might not be equal to $y_i(c_k)$ because exposure levels will be misclassified in the noisy network. We denote the distribution of potential outcomes for the exposure level $c_k$ in the observed network by $\tilde F_{c_k}(x)=\frac{1}{N_v}\sum_{i=1}^{N_v}I_{\{ \tilde y_i(c_k)\leq x \}}$. The proposed bootstrap algorithm proceeds in three main steps:
	\begin{itemize}
		\item[(1)] We compute the empirical cumulative distribution $\hat F_{c_k}(x)=\frac{1}{N_{c_k} }\sum_{i=1}^{N_v}I_{\{f(\bm Z,\tilde{\bm{A}}_{\cdot,i})=c_k\} } I_{\{\tilde Y_i\leq x \}}$ from the individuals for which $f(\bm Z,\tilde{\bm{A}}_{\cdot,i})=c_k$ in the actual experiment, where $N_{c_k}$ is the number of individuals in the exposure level $c_k$ based on the observed network.
		\item[(2)] We then impute potential values $\tilde y_i(c_k)$ for each individual, which is obtained from  the estimated potential outcome distributions.
		\item[(3)] In each iteration, we construct a bootstrap resample matrix $\tilde {\bm A}^b$, assign Bernoulli random of treatment with probability $p$, and obtain the outcomes of each individual from the imputed values $\tilde y_i^b(c_k)$. We then compute our method-of-moments estimators for $\bar y(c_k)$ in the bootstrap sample obtained using the imputed potential outcomes.
	\end{itemize}
	
	Given the simulated method-of-moments estimators, we can estimate the biases and variances of $\tilde{\bar y}_\text{MME}(c_k )$ that are needed to construct confidence intervals. We next describe steps (2) and (3) in detail.

	We impute the missing counterfactuals according to:
	\begin{align}
		\tilde y_i^b(c_k)=\begin{cases}
			\tilde Y_i & \text{ if }f(\bm Z,\tilde{\bm{A}}_{\cdot,i})=c_k \text{ or } \sum_{i=1}^{N_v}I_{\{f(\bm Z,\tilde{\bm{A}}_{\cdot,i})=c_k\} }=0,\\
			\hat F^{-1}_{c_k}\left( \hat F_{f(\bm Z,\tilde{\bm{A}}_{\cdot,i})}(\tilde Y_i) \right)& \text{ otherwise}.
		\end{cases}    
	\end{align}

	For the $b$th bootstrap replication, we construct a bootstrap resample matrix $\tilde {\bm A}^b$ as follows: for entries $\tilde A^b_{i,j}$, $1\leq i<j\leq N_v$, we randomly select one of $m$ observed adjacency matrices, and use the $(i,j)$ entry of the selected matrix as the value of $\tilde A^b_{i,j}$. Then, we set the lower triangular elements equal to the corresponding upper triangular elements and force the diagonal elements to be $0$s. 
	
	We then generate $Z_{1}^b,\cdots,Z_{N_v}^b$ as independent Bernoulli draws with success probability $p$ and obtain the bootstrap sample $Y_i^b\coloneqq\tilde y_i^b(f(\bm Z^b,\tilde {\bm A}^b_{\cdot,i}) )$. Finally, we can compute the bootstrap analogs of the method-of-moments estimators $\tilde{\bar y}_\text{MME}^b(c_k )$.
	
	Repeating the resampling step $B$ times, we obtain a sample $(\tilde{\bar y}_\text{MME}^1(c_k ),\cdots,\tilde{\bar y}_\text{MME}^B(c_k ))$ that can be used to construct bias estimators $\widehat{\text{Bias}} (\tilde{\bar y}_\text{MME}(c_k ) )$ and variance estimators $\widehat{\text{Var}} (\tilde{\bar y}_\text{MME}(c_k ) )$ for tests or confidence intervals. In the homogeneous network setting, by Corollary \ref{coro2}, an approximate 95\% confidence interval for $\bar y(c_k)$ is
	\begin{align}\label{eq3.9}
		\left( \tilde{\bar y}_\text{MME}(c_k )-1.96 \sqrt{\widehat{\text{Var}} (\tilde{\bar y}_\text{MME}(c_k ) ) },\ \tilde{\bar y}_\text{MME}(c_k )+1.96 \sqrt{\widehat{\text{Var}} (\tilde{\bar y}_\text{MME}(c_k ) ) } \right) .
	\end{align}
	In the inhomogeneous network setting, by Theorem \ref{th9}, an approximate 95\% confidence interval for $\bar y(c_k)$ is
	\begin{align}\label{eq3.10}
		\begin{split}
			\bigg( &\tilde{\bar y}_\text{MME}(c_k )-\widehat{\text{Bias}} (\tilde{\bar y}_\text{MME}(c_k ) ) -1.96 \sqrt{\widehat{\text{Var}} (\tilde{\bar y}_\text{MME}(c_k ) ) },\ \tilde{\bar y}_\text{MME}(c_k )-\widehat{\text{Bias}} (\tilde{\bar y}_\text{MME}(c_k ) )\\
			&+1.96 \sqrt{\widehat{\text{Var}} (\tilde{\bar y}_\text{MME}(c_k ) ) } \bigg) .
		\end{split}
	\end{align}
	
	\section{Numerical illustration: British secondary school contact networks} \label{sec5}
	
	We conduct some simulations to illustrate the finite sample properties of the proposed estimation methods.
	We consider the data and network construction described in \cite{kucharski2018structure}.
	These data were collected from 460 unique participants across four rounds of data collection conducted between January and June 2015 in year 7 groups in four UK secondary schools, with 7,315 identifiable contacts reported in total.
	They used a process of peer nomination as a method for data collection: students were asked, via the research questionnaire, to list the six other students in year 7 at their school that they spend the most time with.
	For each pair of participants in a specific round of data collection, a single link was defined if either one of the participants reported a contact between the pair (i.e.
	there was at least one unidirectional link, in either direction).
	Our analysis focuses on the single link contact network.
	
	For each school, we construct a `true' adjacency matrix $\bm A$: if an edge occurs between a pair of vertices more than once in four rounds, we view that pair to have a true edge.
	The noisy, observed adjacency matrices $\tilde{\bm A}$, $\tilde{\bm A}_*$, $\tilde{\bm A}_{**}$ are generated according to (\ref{eq3.1}).
	We set $\alpha=0.005$ or 0.010, and $\beta=0.05,\ 0.10$, or 0.15.
	We assume that both $\alpha$ and $\beta$ are unknown.
	For treatment effects we adopt a simple model in the spirit of the ‘dilated effects’ model of Rosenbaum \citep{rosenbaum1999reduced} and suppose $y_i(c_{11})=10,\ y_i(c_{10})=7,\ y_i(c_{01})=5,\ y_i(c_{00})=1$.
	We set $p=0.1$ and explore the effect of $\alpha,\ \beta$ on the performance of estimators $\tilde{\bar y}(\cdot)$.

	We run Monte Carlo simulation of 10,000 trials and compute three kinds of estimators: Aronow and Samii estimators in noise-free networks, Aronow and Samii estimators in noisy networks, and method-of-moments estimators in noisy networks.
	For the method-of-moments estimators, we first obtain estimators $\hat\alpha$ and $\hat\beta$ by Algorithm \ref{algo1}.
	The networks are sparse with small sizes, so we compute $\tilde {\overline{{y}}}_\text{MME}(c_k)$ by (\ref{eq:mme2}).
	Also, we compute the conservative estimator of $\text{Var}(\hat {\overline{{y}}}_\text{A\&S}(c_k))$ defined in \cite{aronow2017estimating}, and apply the bootstrap algorithm presented in Section \ref{sec3.3} to obtain the bias and variance estimators for $\tilde {\overline{ {y}}}_\text{A\&S}(c_k)$ and $\tilde {\overline{ {y}}}_\text{MME}(c_k)$. Then we construct $95\%$ confidence intervals for three kinds of estimators and report the coverage rates. The network degrees in Schools 1, 2 and 3 are closer to Poisson distributions compared to power law distributions in terms of Akaike information criterion, bayesian information criteria and Kolmogorov–Smirnov statistic, while the network degree distribution of School 4 is closer to a power law distribution. Therefore, we use (\ref{eq3.9}) to construct confidence intervals for our method-of-moments estimators in Schools 1, 2 and 3, and use (\ref{eq3.10}) for School 4.  Figure \ref{fig1} shows these results for the edge error settings in Table \ref{tab3}.
	\begin{table}[!ht]
		\caption{Edge error settings in the simulation study.}
		\begin{center}
			\begin{tabular}{l|llllll}
				\hline  
				& case 1 & case 2&  case 3 &  case 4&  case 5 &  case 6\\\hline
				$\alpha$ & $0.005$ & $0.005$  & $0.005$ & $0.01$ & $0.01$ & $0.01$ \\
				$\beta$ & $0.05$  & $0.1$   & $0.15$ & $0.05$  & $0.1$   & $0.15$ \\
				\hline
			\end{tabular}
		\end{center}
		\label{tab3}
	\end{table} 
	
	\begin{figure} 
		\centering
		\caption{Biases, standard deviations, mean absolute errors of standard errors, and coverage rates of 95\% confidence intervals in four schools.}	
		\includegraphics[width=\textwidth]{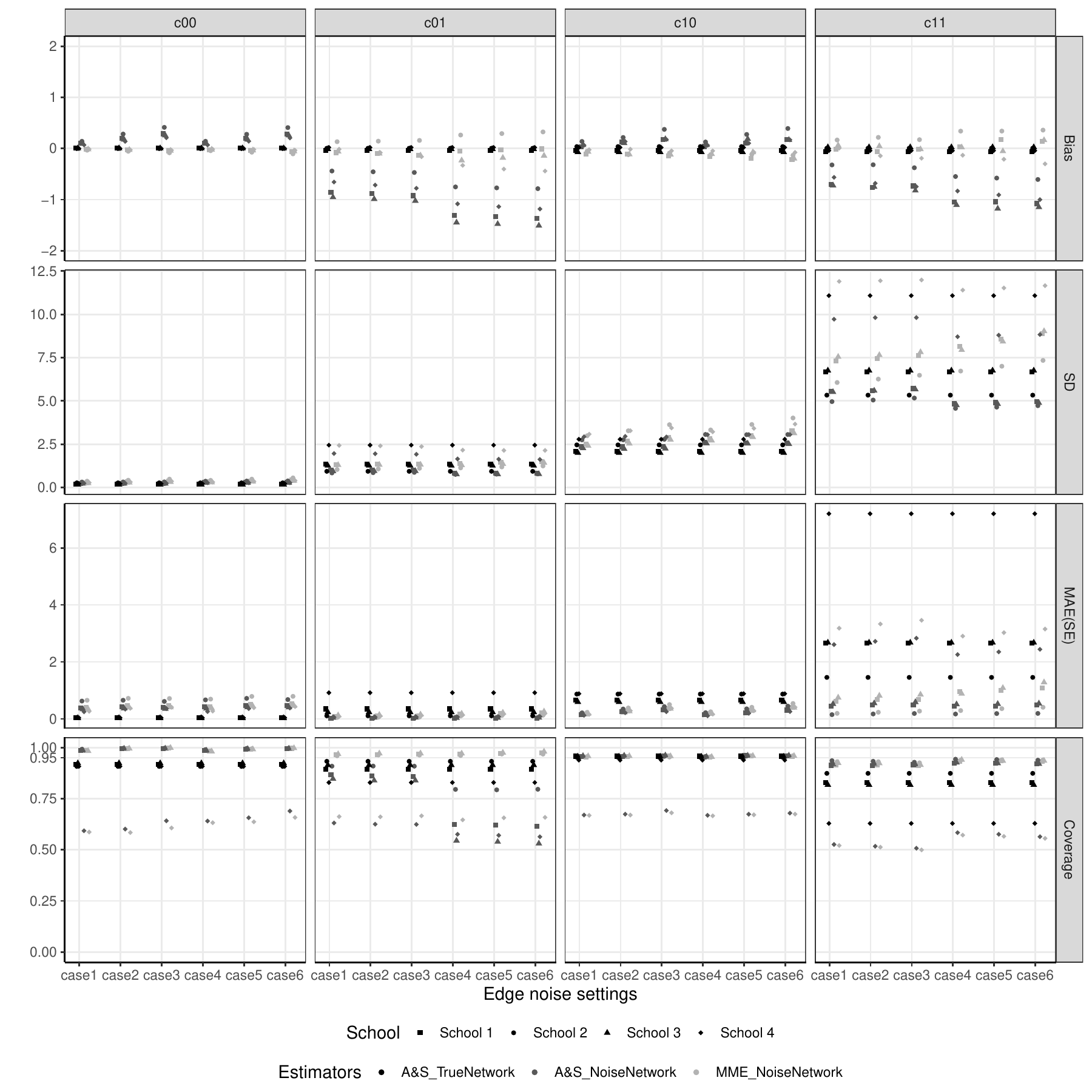}
		\label{fig1}
	\end{figure}
	
	From the plots, we see that method-of-moments estimators outperform Aronow and Samii estimators in noisy networks, and essentially perform the same on noisy networks as Aronow and Samii estimators do on noise-free networks (same zero biases with, at times, just slightly larger standard deviations). Aronow and Samii estimators in noisy networks underestimate $\bar y (c_{11})$ and $\bar y (c_{01})$ and overestimate $\bar y (c_{10})$ and $\bar y (c_{00})$.
	And the biases of Aronow and Samii estimators for $\bar y (c_{11})$ and $\bar y (c_{01})$ increase as $\alpha$ and $\beta$ increase, while the biases of Aronow and Samii estimators for $\bar y (c_{01})$ and $\bar y (c_{00})$ only depend on $\beta$.
	The biases of method-of-moments estimators are close to zero in all cases.
	In addition, standard deviations of the three types estimators are similar in all cases.
	The standard deviations of estimators in School 4 are larger than those in other schools because the network size in School 4 is relatively small. 
	
	In addition, the standard errors of method-of-moments estimators are close to the corresponding standard deviations expect for the exposure level $c_{11}$ in School 4. The network size in School 4 is small, so there is almost no individual in exposure level $c_{11}$. Thus, we are unable to impute missing counterfactuals accurately, which leads to inaccurate variance estimators. Furthermore, the variances of bias estimate are large because of the small network size in School 4. Therefore, we cannot obtain good bias estimators. The inaccurate bias and variance estimators in School 4 lead to relatively low coverage rates for method-of-moments estimators.
	
	\section{Discussion}\label{sec6}
	
	
	Here we have quantified biases and variances of standard estimators in noisy networks and developed a general framework for estimation of true average causal effects in contexts wherein one has observations of noisy networks.
	Our approach requires knowledge or consistent estimates of the corresponding noise parameters, the latter which can be obtained with as few as three replicates of network observations.  We employ method-of-moments techniques to derive estimators and establish their asymptotic unbiasedness, consistency, and normality.
	Simulations in British secondary schools contact networks demonstrate that substantial inferential accuracy by method-of-moments estimators is possible in networks of even modest size when nontrivial noise is present.

	We have pursued a frequentist approach to the problem of uncertainty quantification for estimating average causal effects.
	If the replicates necessary for our approach are unavailable in a given setting, a Bayesian approach is a natural alternative.
	For example, posterior-predictive checks for goodness-of-fit based on examination of a handful of network summary measures is common practice (e.g., \cite{bloem2018random}).
	Note, however, that the Bayesian approach requires careful modeling of the generative process underlying $G$ and typically does not distinguish between signal and noise components.
	Our analysis is conditional on $G$, and hence does not require that $G$ be modeled.
	It is effectively a `signal plus noise’ model, with the signal taken to be fixed but unknown.
	Related work has been done in the context of graphon modeling, with the goal of estimating network motif frequencies (e.g., \cite{latouche2016variational}).
	However, again, one typically does not distinguish between signal and noise components in this setting.
	Additionally, we note that the problem of practical graphon estimation itself is still a developing area of research.
	
	Our work here sets the stage for extensions to other potential outcome frameworks and exposure models. Here we sketch the key elements of one such extension.
	For example, consider the exposure mapping $f$ as following:
	\begin{equation} \label{eq6.2}
		f(\bm z,\bm A_{\cdot i})=\begin{cases}
			c_{11'} \text{(Direct + $\geq$ $m_i$ Neighborhood Exposure)}  , & z_i I_{\{\bm z^\top \bm A_{\cdot i}>0\}}\geq m_i,\\
			c_{10'} \text{(Direct + $<$ $m_i$ Neighborhood Exposure)} , & z_i I_{\{\bm z^\top \bm A_{\cdot i}>0\}}<m_i,\\
			c_{01'} \text{( $\geq m_i$ Neighborhood Exposure)}  , & (1-z_i) I_{\{\bm z^\top \bm A_{\cdot i}>0\}}\geq m_i,\\
			c_{00'}  \text{( $< m_i$ Neighborhood Exposure)} , & (1-z_i) I_{\{\bm z^\top \bm A_{\cdot i}>0\}}<m_i,\\
		\end{cases}
	\end{equation}
	where $m_i\geq 1$. 
	When $m_i=1$, it reduces to (\ref{eq2.1}). And if $m_i=k$, then the level $c_{11'}$ is known as the absolute $k$-neighborhood exposure (\cite{ugander2013graph}). When $m_i=q d_i$, $0\leq q\leq 1$, the level $c_{11'}$ is  called the fractional $q$-neighborhood exposure (\cite{ugander2013graph}).
	The generalized four-level exposure model provides useful abstractions for the analysis of networked experiments.
	For example, infectious diseases (e.g., like COVID-19) are more likely to spread between people closely connected in a social network.
	And being in contact with more people with the disease means that, in theory, they are more likely to contract the disease.
	
	As an illustration, suppose that treatment is assigned to the $N_v$ individuals in a network through Bernoulli random sampling, with probability $p$. The exposure probabilities for four levels can be found in supplementary material D.
	Next, we show orders of the exposure probabilities for nodes with varying degrees in Theorem \ref{th7}.
	If $d_i=\Theta(1/p)$, upper bounds for four exposure probabilities do not depend on $m_i$.
	When $d_i=\omega(1/p)$, upper bounds for $p_i^e(c_{10'})$ and $p_i^e(c_{00'})$ increase as $m_i$ increases.
	And upper bounds for $p_i^e(c_{11'})$ and $p_i^e(c_{01'})$ decrease as $m_i$ increases if $d_i=o(1/p)$. See supplementary material D for the proof.
	
	\begin{theorem}\label{th7}
		Assume a generalized four-level exposure model and Bernoulli random assignment of treatment with $p=o(1)$.
		And for all $i$, $d_i\geq m_i$.
		Then, for all integers $m_i\geq 1$ and $m_i=\mathcal O(1)$, the orders of exposure probabilities are as follows.
		
		\begin{table}[!h]\label{tab2}
			\begin{center}
				\begin{tabular}{l|llll}
					& $p_i^e(c_{11'})$ &  $p_i^e(c_{10'})$&  $p_i^e(c_{01'})$ &  $p_i^e(c_{00'})$\\\hline
					$d_i=\omega(1/p)$ & $\mathcal O(p)$ & $\mathcal O( p (d_ip)^{m_i-1}/e^{d_ip} )$ & $\mathcal O(1)$& $\mathcal O((d_ip)^{m_i-1}/e^{d_ip} )$ \\
					$d_i=\Theta(1/p)$ & $\mathcal O( p)$ & $\mathcal O( p)$  & $\mathcal O(1)$ & $\mathcal O(1)$ \\
					$d_i=o(1/p)$ & $\mathcal O( p(d_ip)^{m_i} )$  & $\mathcal O( p)$   & $\mathcal O( (d_ip)^{m_i})$ & $\mathcal O(1)$ \\
				\end{tabular}
			\end{center}
		\end{table}
	\end{theorem}
	Then, we can construct regularity conditions for the average causal effect estimators to be consistent.
	Similarly, one can quantify biases and variances of standard estimators in noisy networks and develop a general framework for estimation of true average causal effects.
	These require additional work due to the complexities of formulas for exposure probabilities.
	
	Our choice to work with independent network noise is both natural and motivated by convenience.
	A precise characterization of the noise dependency would be needed to extend our work, but is typically problem-specific and hence a topic for further investigation.
	
	\section{Supplementary Materials}
	
	\textbf{Supplementary Materials for ``Causal Inference under Network Interference with Noise"}: Providing proofs of all propositions and theorems presented in the main paper.
	
	\textbf{Data and code accessibility:} No primary data are used in this paper. Secondary data source is taken from \cite{kucharski2018structure}. These data and the code necessary to reproduce the results in this paper are available at \url{https://github.com/KolaczykResearch/CausInfNoisyNet}.
	
	\section{Acknowledgement}
	
	This work was supported in part by ARO award W911NF1810237. This work was also supported by the Air Force Research Laboratory and DARPA under agreement number FA8750-18-2-0066 and by a grant from MIT Lincoln Labs.
	
	\section{Appendix}
	
	In this appendix, we provide arguments for the consistency of contrast estimates in noise-free networks and regularity conditions in noisy networks. 
	
	\subsection{Consistency of contrast estimates in noise-free networks}\label{app1}
	
	We first establish conditions for the estimator $\hat{\uptau}(c_k,c_l)$ to converge to $\uptau(c_k,c_l)$ as $N_v\rightarrow\infty$.
	We will show that, under two regularity conditions, $\hat{\uptau}(c_k,c_l)\xrightarrow{P} \uptau(c_k,c_l)$ as $N_v\rightarrow\infty$. Note that these conditions are similar to but slightly more general than the conditions in \cite{aronow2017estimating}.
	
	\begin{condition}\label{cond1}
		For all values $i$ and $c_k$, $|y_i(c_k)|\leq c<\infty$, $p^e_i(c_k)>0$ and $\sum_{i=1}^{N_v}1/p^e_i(c_k)=o(N_v^2)$, where $c$ is a constant.
	\end{condition} 
	
	We will also make an assumption about the amount of dependence among exposure conditions in the population. Let $ p_{ij}^e(c_k)=\sum_{\bm z}p_{\bm z}I_{\{f({\bm z}, {\bm x}_i)=c_k\}}I_{\{f({\bm z}, {\bm x}_j)=c_k\}}$.

	\begin{condition}\label{cond2}
		For all values $c_k$, $\sum_{i=1}^{N_v}\sum_{j\neq i}^{N_v}| p^e_{ij}(c_k)/(p^e_i(c_k)p^e_j(c_k))-1 |=o(N_v^2)$.
	\end{condition}	
	
	Condition \ref{cond2} implies that the amount of pairwise clustering in exposure conditions is limited in scope as $N_v$ grows.
	Condition \ref{cond2} can be relaxed, though Condition \ref{cond1} would likely need to be strengthened accordingly.
	
	\begin{proposition}
		Given Conditions \ref{cond1} and \ref{cond2}, $\hat{\uptau}(c_k,c_l)\xrightarrow{P} \uptau(c_k,c_l)$ as $N_v\rightarrow\infty$.
	\end{proposition}
	
	Assuming the four-level exposure model in (\ref{eq2.1}) and Bernoulli random assignment of treatment with probability $p$, we consider the consistency of the estimator $\hat{\uptau}(c_k,c_l)$ in two typical classes of networks: homogeneous and inhomogeneous.
	
	\begin{proposition}[Homogeneous]\label{pro1}
		Assume a four-level exposure model and Bernoulli random assignment of treatment with $p$.
		In the homogeneous network setting, under Assumption \ref{a4}, $\hat{\uptau}(c_k,c_l)\xrightarrow{P} \uptau(c_k,c_l)$ as $N_v\rightarrow\infty$.
	\end{proposition}
	
	\begin{proposition}[Inhomogeneous]\label{pro2}
		Assume a four-level exposure model and Bernoulli random assignment of treatment with $p$.
		In the inhomogeneous network setting, under Assumption \ref{a4}, $\lambda=\Theta(p)$ and $\lambda>p$, we have $\hat{\uptau}(c_k,c_l)\xrightarrow{P} \uptau(c_k,c_l)$ as $N_v\rightarrow\infty$.
	\end{proposition}  
	
	Proofs for Propositions 1 -- 3  appear in the supplementary material A.

	Note that, under Assumption \ref{a4}, Condition \ref{cond1}  does not hold for levels $c_{10}$ and $c_{00}$ when the degrees follow a Pareto distribution with shape $\zeta>1$.
	This is because there are more high degree nodes, and $1/p_i^e(c_{10})$ and $1/p_i^e(c_{00})$ increase exponentially when the degree $d_i$ increases.
	See supplementary material E for the proof.
	
	\subsection{Standard estimators in noisy networks}\label{app2}

	Recall that under the Condition \ref{cond1} and \ref{cond2}, $\hat{\uptau}(c_k,c_l)\xrightarrow{P} \uptau(c_k,c_l)$ as $N_v\rightarrow\infty$.
	By making assumptions on underlying rates of error $\alpha$ and $\beta$, we will show that similar regularity conditions hold for noisy homogeneous and inhomogeneous networks.
	These conditions will then be used in our characterization of bias and variance in Sections \ref{sec3.1} and \ref{sec3.2}.
	Define $\tilde p_{ij}^e(c_k)=\sum_{\bm z}p_{\bm z}I_{\{f({\bm z}, \tilde{\bm X}_i)=c_k\}}I_{\{f({\bm z}, \tilde{\bm X}_j)=c_k\}}$.

	\begin{proposition}[Homogeneous]\label{pro3}
		Assume a four-level exposure model and Bernoulli random assignment of treatment with $p$.
		In the homogeneous network setting, under Assumptions \ref{a1} - \ref{a5}, for all values $i$ and $c_k$, $\mathbb P(\tilde p^e_i(c_k)>0)\rightarrow 1$, $\mathbb E[ \sum_{i=1}^{N_v}I_{\{\tilde p^e_i(c_k)>0\} }/\tilde p^e_i(c_k)]=o(N_v^2)$, and $\mathbb E[\sum_{i=1}^{N_v}\sum_{j\neq i}^{N_v}I_{\{\tilde p^e_i(c_k)>0\} }I_{\{\tilde p^e_j(c_k)>0\} }| \tilde p^e_{ij}(c_k)/(\tilde p^e_i(c_k)\tilde p^e_j(c_k))-1 |]=o(N_v^2)$.
	\end{proposition}	
	
	\begin{proposition}[Inhomogeneous]\label{pro4}
		Assume a four-level exposure model and Bernoulli random assignment of treatment with $p$.
		In the inhomogeneous network setting, under Assumptions \ref{a1}- \ref{a5}, $\lambda=\Theta(p)$ and $\lambda>p$,  the statements in Proposition \ref{pro3} hold for all values $i$ and $c_k$.
	\end{proposition} 
	See supplementary material B for proofs of Propositions 4 and 5.

\end{document}


\def\spacingset#1{\renewcommand{\baselinestretch}%
		{#1}\small\normalsize} \spacingset{1}

	
	\if1\blind
	{
		\title{\bf Supplementary Materials for ``Causal Inference under Network Interference with Noise"}
		\author{Wenrui Li\thanks{\textit{Contact:} Wenrui Li, wenruili@bu.edu,
				Department of Mathematics and Statistics, Boston University, 111 Cummington Mall, Boston, MA 02215, USA.}\\
			Department of Mathematics \& Statistics, Boston University\\
			Daniel L. Sussman \\
			Department of Mathematics \& Statistics, Boston University\\
			and \\
			Eric D. Kolaczyk\hspace{.2cm}\\
			Department of Mathematics \& Statistics, Boston University}
		\maketitle
	} \fi
	
	\if0\blind
	{
		\bigskip
		\bigskip
		\bigskip
		\begin{center}
			{\LARGE\bf Supplementary Materials for ``Causal Inference under Network Interference with Noise"}
		\end{center}
		\medskip
	} \fi

	\bigskip

	\noindent%
	\textbf{Summary:} In this Supplementary Materials document, we provide proofs of all propositions and theorems presented in the main paper.
	\vfill

	\appendix
	
	\newpage
	\spacingset{1.5} 
	\section{Proofs of propositions for noise-free networks}
	
	\subsection{Proof of Proposition 1}
	
	Notice that
	
	\begin{align*}
	\text{Var}\Big[ \hat{\overline{y}}(c_k)\Big]=&\ \frac{1}{N_v^2}\Bigg\{\sum_{i=1 }^{N_v}p_i^e(c_k)\big [1-p_i^e(c_k)\big]\Big[\frac{y_i(c_k)}{p_i^e(c_k)}\Big]^2\\
	&+ \sum_{i=1}^{N_v} \sum_{j\neq i } \big [p_{ij}^e(c_k)-p_i^e(c_k)p_j^e(c_k)\big]\frac{y_i(c_k)}{p_i^e(c_k)}\frac{y_j(c_k)}{p_j^e(c_k)}\Bigg\}.
	\end{align*}
	Thus, under Condition 1 and 2, we have $\text{Var}[ \hat{\overline{y}}(c_k)]=o(1)$ and $\text{Var}[ \hat{\overline{y}}(c_l)]=o(1)$. 
	
	Then, by Cauchy-Schwarz inequality, we obtain
	\begin{align*}
	\text{Var} \Big[\hat{\uptau}(c_k,c_l)\Big]&= \text{Var}\Big[ \hat{\overline{y}}(c_k)\Big]+ \text{Var}\Big[ \hat{\overline{y}}(c_l)\Big]-2\ \text{Cov}\Big[ \hat{\overline{y}}(c_k),\hat{\overline{y}}(c_l)\Big]\\
	&\leq \text{Var}\Big[ \hat{\overline{y}}(c_k)\Big]+ \text{Var}\Big[ \hat{\overline{y}}(c_l)\Big]+2\sqrt{  \text{Var}\Big[ \hat{\overline{y}}(c_k)\Big]\text{Var}\Big[ \hat{\overline{y}}(c_l)\Big]}.
	\end{align*}
	Thus, we obtain $\text{Var} [\hat{\uptau}(c_k,c_l)]=o(1)$. Since $\mathbb E [\hat{\uptau}(c_k,c_l)]=\uptau(c_k,c_l)$, we have $\hat{\uptau}(c_k,c_l)\xrightarrow{L_2}\uptau(c_k,c_l)$. This implies $\hat{\uptau}(c_k,c_l)\xrightarrow{P}\uptau(c_k,c_l)$ as $N_v\rightarrow\infty$.
	
	\subsection{Proof of Proposition 2}
	
	By Proposition 1, it suffices to show Condition 1 and 2 are satisfied. 
	For notational simplicity, we define 
	$\mathbb E[1/p^e(c_k)]\coloneqq \frac{1}{N_v}\sum_{i=1}^{N_v}1/p^e_i(c_k)$, i.e., the average of the inverse of the exposure probability over a finite population $N_v$.
	
	Under Assumption 4, we obtain
	\begin{align*}
	\mathbb E[1/p^e(c_{10})]&=\frac{e^{\bar d /(1-p)}-1 }{(e^{\bar d} -1)p}=\Theta(1/p),\\
	\mathbb E[1/p^e(c_{00})]&=\frac{e^{\bar d /(1-p)}-1 }{(e^{\bar d} -1)(1-p)}=\mathcal O(1).
	\end{align*}	
	
	Note that $\text{Var}[1-(1-p)^d]=o(1)$, where $d$ is the degree variable. Thus, $1-(1-p)^d\xrightarrow{L_2} \lim_{N_v\rightarrow\infty}\mathbb E[ 1-(1-p)^d]$. This implies $1-(1-p)^d\xrightarrow{P} \lim_{N_v\rightarrow\infty}\mathbb E[ 1-(1-p)^d]$. By the continuous mapping theorem, we have $\dfrac{p}{1-(1-p)^d }\xrightarrow{P}\lim_{N_v\rightarrow\infty}\dfrac{p}{\mathbb E[ 1-(1-p)^d]}$. Thus, we obtain
	\begin{align*}
	\lim_{N_v\rightarrow\infty} \mathbb E\Big[	\dfrac{p}{1-(1-p)^d }\Big] =\lim_{N_v\rightarrow\infty}\dfrac{p}{\mathbb E[ 1-(1-p)^d]}.
	\end{align*}
	Since $\mathbb E[ 1-(1-p)^d]=\mathcal O(1)$, we have $\mathbb E[1/p^e(c_{11})]=\Theta(1/p)$ and $\mathbb E[1/p^e(c_{01})]=\mathcal O(1)$. Therefore, Condition 1 is satisfied. 
	
	Next, we show Condition 2 is also satisfied. Define a pairwise dependency indicator $g_{ij}$ such that if $g_{ij}=0$, then $f({\bm Z}, {\bm A}_{\cdot i} )\independent f({\bm Z}, {\bm A}_{\cdot j})$, else let $g_{ij}=1$. Note that $g_{ij}=1$ if $C_{ij}\geq 1$ or $A_{i,j}=1$. Then, we have $\sum_{i=1}^{N_v}\sum_{j\neq i}^{N_v}g_{ij}=o(N_v^2)$ because $\sum_{i=1}^{N_v}\sum_{j\neq i}^{N_v} I_{\{C_{ij}=0\}} \sim N_v^2$ and $\sum_{i=1}^{N_v}\sum_{j\neq i}^{N_v} I_{\{A_{i,j}=0 \}}\sim N_v^2$. Notice that 
	\begin{align*}
	\sum_{i=1}^{N_v}\sum_{j\neq i}^{N_v} | p^e_{ij}(c_k)/(p^e_i(c_k)p^e_j(c_k))-1 |=&
	\sum_{i=1}^{N_v}\sum_{j\neq i}^{N_v}g_{ij}| p^e_{ij}(c_k)/(p^e_i(c_k)p^e_j(c_k))-1 |\\
	\leq &\sum_{i=1}^{N_v}\sum_{j\neq i}^{N_v}g_{ij}p^e_{ij}(c_k)/(p^e_i(c_k)p^e_j(c_k))+\sum_{i=1}^{N_v}\sum_{j\neq i}^{N_v}g_{ij}.
	\end{align*}
	It suffices to show 
	\begin{align}\label{eq1}
	\sum_{i=1}^{N_v}\sum_{j\neq i}^{N_v}\frac{g_{ij} p^e_{ij}(c_k)}{p^e_i(c_k)p^e_j(c_k)}=o(N_v^2).
	\end{align}
	By Young's inequality, we have
	\begin{align*}
	\sum_{i=1}^{N_v}\sum_{j\neq i}^{N_v}\frac{g_{ij} p^e_{ij}(c_k)}{p^e_i(c_k)p^e_j(c_k)}\leq \Bigg( \sum_{i=1}^{N_v}\sum_{j\neq i}^{N_v} g_{ij}\Bigg)^{1/2}\Bigg( \sum_{i=1}^{N_v}\sum_{j\neq i}^{N_v} \Big[ \frac{ p^e_{ij}(c_k)}{p^e_i(c_k)p^e_j(c_k)}\Big]^2\Bigg)^{1/2}.
	\end{align*}
	For the level $c_{00}$, we have 
	\begin{align*}
	\sum_{i=1}^{N_v}\sum_{j\neq i}^{N_v} \Big[ \frac{ p^e_{ij}(c_{00})}{p^e_i(c_{00})p^e_j(c_{00})}\Big]^2\leq  \sum_{i=1}^{N_v}\sum_{j\neq i}^{N_v} \Big[ \frac{ 1}{p^e_i(c_{00})p^e_j(c_{00})}\Big]^2\leq  \Bigg( \sum_{i=1}^{N_v} \Big[ \frac{ 1}{p^e_i(c_{00})}\Big]^2 \Bigg)^2
	\end{align*}
	and $\mathbb E [1/(p^e(c_{00}))^2 ]=\mathcal O(1)$. Thus, we obtain 
	\begin{align*}
	\sum_{i=1}^{N_v}\sum_{j\neq i}^{N_v}| p^e_{ij}(c_{00})/(p^e_i(c_{00})p^e_j(c_{00}))-1 |=o(N_v^2).
	\end{align*}
	Similarly, we can show (\ref{eq1}) for other exposure levels. Therefore, Condition 2 is satisfied. 
	
	\subsection{Proof of Proposition 3}
	First, we compute the normalization parameter of the Pareto distribution with an exponential cutoff. By definitions, we have
	\begin{align*}
	C(\zeta, d_L, \lambda ) \int_{d_L}^{N_v-1} e^{- \lambda x} x^{-(\zeta+1)}dx &=1, \\
	C(\zeta, d_L, \lambda) \int_{d_L}^{N_v-1} x\cdot e^{- \lambda x } x^{-(\zeta+1)}dx &=\bar d.
	\end{align*}
	Note that, by integration by parts, we have
	\begin{align*}
	\int_{d_L}^{N_v-1} x\cdot e^{- \lambda x} x^{-(\zeta+1)}dx =&\ \frac{1}{\lambda}\cdot \big( d_L^{-\zeta} \cdot e^{-\lambda d_L}-(N_v-1)^{-\zeta} \cdot e^{-\lambda(N_v-1) }\big)\\
	&-\frac{\zeta  }{\lambda} \cdot\int_{d_L}^{N_v-1} e^{- \lambda x } x^{-(\zeta+1)}dx .
	\end{align*}
	Therefore, we obtain 
	\begin{align}
	C(\zeta, d_L, \lambda) &=(\lambda\bar d+\zeta) \big/\big(d_L^{-\zeta} \cdot e^{-\lambda d_L } -(N_v-1)^{-\zeta} \cdot e^{-\lambda (N_v-1) }\big) \label{eq:A2} ,\\
	\int_{d_L}^{N_v-1} e^{- \lambda x} x^{-(\zeta+1)}dx  &=  \big(d_L^{-\zeta} \cdot e^{-\lambda d_L } -(N_v-1)^{-\zeta} \cdot e^{-\lambda (N_v-1) }\big) \big/(\lambda\bar d+\zeta)\label{eq:A3} .
	\end{align}
	
	Next, we show $d_L=\Theta(\bar d)$. By definitions of $d_L$ and $\bar d$, we have $d_L=\mathcal O(\bar d)$. Therefore, it suffices to show that $d_L= \Omega(\bar d)$. We prove it by contradiction. Assume $d_L= o(\bar d)$, then by (\ref{eq:A3}), we have
	\begin{align}
	\int_{d_L}^{N_v-1} e^{- \lambda x} x^{-(\zeta+1)}dx  \sim d_L^{-\zeta} /(\lambda\bar d+\zeta).\label{eq:A4} 
	\end{align}
	On the other hand, note that  
	\begin{align}
	\int_{d_L}^{N_v-1} e^{- \lambda x} x^{-(\zeta+1)}dx =\lambda^\zeta\bigg\{\Gamma\big(-\zeta,\lambda d_L \big) -  \Gamma\big(-\zeta,\lambda (N_v-1) \big) \bigg\}\label{eq:A5} 
	\end{align}
	where $\Gamma(\cdot,\cdot)$ is the upper incomplete gamma function. Recall the properties of the upper incomplete gamma function, $\Gamma(s,x)\rightarrow -x^s/s$ as $x\rightarrow 0$ and $s<0$, and $\Gamma(s,x)\rightarrow x^{s-1}e^{-x}$ as $x\rightarrow \infty$ (\cite{jameson2016incomplete}, \cite{olver1997asymptotics}, \cite{temme2011special}). Then by  (\ref{eq:A5}), we obtain 
	\begin{align}
	\int_{d_L}^{N_v-1} e^{- \lambda x} x^{-(\zeta+1)}dx \sim d_L^{-\zeta} / \zeta ,
	\end{align}
	which is in contradiction to (\ref{eq:A4}). Thus, we have $d_L=\Theta(\bar d)$.

	Then, we prove Proposition 3. By Proposition 1, it suffices to show Condition 1 and 2 are satisfied.  
	Under Assumption 4 and $p<\lambda$, we obtain
	\begin{align*}
	\mathbb E[1/p^e(c_{10})]& =\frac{1}{p}\cdot C(\zeta, d_L, \lambda) \cdot \int_{d_L}^{N_v-1} (1-p)^{-x} e^{- \lambda x} x^{-(\zeta+1)}dx  \\
	& =\mathcal O\Bigg(  \frac{1}{p}\cdot C(\zeta, d_L, \lambda) \cdot d_L^{-(\zeta+1)} \cdot \int_{d_L}^{N_v-1} (1-p)^{-x} e^{- \lambda x} dx \Bigg)  \\
	&=\mathcal O\Bigg(\frac{1 }{pd_L[\lambda+ \log(1-p)  ]} \Big\{\Big[\frac{1}{e^{\lambda}(1-p)}\Big]^{d_L}-\Big[\frac{1}{e^{\lambda}(1-p)}\Big]^{N_v}\Big\} \Bigg) \\
	&=\mathcal O(1/p). 
	\end{align*}	
	Similarly, we have $\mathbb E[1/p^e(c_{00})]=\mathcal O(1)$.
	
	Next, we show $\mathbb E[1/p^e(c_{11})]=\mathcal O(1/p)$ and $\mathbb E[1/p^e(c_{01})]=\mathcal O(1)$. Note that 
	\begin{align*}
	\mathbb E[1/p^e(c_{11})]& =\frac{1}{p}\cdot C(\zeta, d_L, \lambda) \cdot \int_{d_L}^{N_v-1} \frac{1}{1-(1-p)^x } \ e^{- \lambda x} x^{-(\zeta+1)}dx  \\
	& =\mathcal O\Bigg(  \frac{1}{p}\cdot C(\zeta, d_L, \lambda) \cdot d_L^{-(\zeta+1)} \cdot \int_{d_L}^{N_v-1} \frac{1}{1-(1-p)^x } \ e^{- \lambda x} dx \Bigg)  \\
	& =\mathcal O\Bigg(  \frac{1}{p}\cdot C(\zeta, d_L, \lambda) \cdot d_L^{-(\zeta+1)} \cdot \int_{d_L}^{N_v-1}  e^{- \lambda x} dx \Bigg)  \\
	& =\mathcal O\Bigg(  \frac{1}{p}\cdot C(\zeta, d_L, \lambda) \cdot d_L^{-(\zeta+2)}    \Bigg)  \\
	&=\mathcal O(1/p). 
	\end{align*}
	Similarly, we have $\mathbb E[1/p^e(c_{01})]=\mathcal O(1)$.
	
	Analogous to the proof of Proposition 2, we can show Condition 2 is also satisfied.   
	
	\section{Proofs of propositions for noisy networks}
	
	\subsection{Proof of Proposition 4}
	In the observed network, the four exposure probabilities become:
	\begin{align*}
	\tilde p_i^e(c_{11})&=p\big \{1-(1-p)^{\tilde d_i} \big \}, \\
	\tilde p_i^e(c_{10})&=p(1-p)^{\tilde d_i},\\
	\tilde p_i^e(c_{01})&=(1-p) \big \{1-(1-p)^{\tilde d_i} \big \}, \\
	\tilde p_i^e(c_{00})&=(1-p)^{\tilde d_i+1},
	\end{align*}
	where $\tilde d_i$ is the degree of vertex $i$ in the observed network. 
	
	(i) For all values $i$ and $c_k$, $\mathbb P(\tilde p^e_i(c_k)>0)\rightarrow 1$ as $N_v\rightarrow\infty$.
	
	Note that $\mathbb P(\tilde p^e_i(c_{10})>0)=\mathbb P(\tilde p^e_i(c_{00})>0)=1$ and $\mathbb P(\tilde p^e_i(c_{11})>0)=\mathbb P(\tilde p^e_i(c_{01})>0)=\mathbb P(\tilde d_i>0)$. Under Assumption 1 - 5, we have
	\begin{align*}
	\lim_{N_v\rightarrow\infty}\mathbb P(\tilde d_i>0)=\lim_{N_v\rightarrow\infty}1-(1-\alpha)^{N_v-1-d_i}\beta^{d_i}=1.
	\end{align*}
	Thus, for all values $i$ and $c_k$, we have $\mathbb P(\tilde p^e_i(c_k)>0)\rightarrow 1$ as $N_v\rightarrow\infty$.
	
	(ii) For all $c_k$, $\mathbb E[ \sum_{i=1}^{N_v}I_{\{\tilde p^e_i(c_k)>0\} }/\tilde p^e_i(c_k)]=o(N_v^2)$.
	
	Define $\check d_i=\sum_{j=1}^{N_v} \tilde{A}_{j,i}A_{j,i}$. We note that $\tilde d_i=(\tilde d_i-\check d_i)+\check d_i$, where $\check d_i$ and $\tilde d_i-\check d_i$ are two independent binomial random variables. Thus, we have $\tilde d_i\sim \text{Binomial}(N_v-1-d_i,\alpha ) + \text{Binomial}(d_i,1-\beta )$. Then, we obtain
	\begin{align*}
	\mathbb{E}[(1-p)^{\tilde d_i}]=&\ (1-\alpha p)^{N_v-1-d_i}\big[1-(1-\beta)p\big]^{d_i},\\
	\text{Var}[(1-p)^{\tilde d_i}]=&\ \big[1-\alpha p(2-p)\big]^{N_v-1-d_i}\big[1-(1-\beta)p(2-p)\big]^{d_i}\\
	&-(1-\alpha p)^{2(N_v-1-d_i)}\big[1-(1-\beta)p\big]^{2d_i},\\
	\mathbb{E}[(1-p)^{-\tilde d_i}]&=\Big(1+\frac{\alpha p}{1-p}\Big)^{N_v-1-d_i}\Big(1+\frac{(1-\beta)p}{1-p}\Big)^{d_i}.
	\end{align*}	
	Since $\mathbb{E}[(1-p)^{-\tilde d_i}]=\mathcal O((1-p)^{- d_i})$, we have $\mathbb E[1/\tilde p_i^e(c_{10})]=\mathcal O(1/ p_i^e(c_{10}))$ and $\mathbb E[1/\tilde p_i^e(c_{00})]=\mathcal O(1/ p_i^e(c_{00}))$. Therefore, by Proposition 2, we obtain 
	$\mathbb E[ \sum_{i=1}^{N_v}I_{\{\tilde p^e_i(c_{10})>0\} }/\tilde p^e_i(c_{10})]=o(N_v^2)$ and $\mathbb E[ \sum_{i=1}^{N_v}I_{\{\tilde p^e_i(c_{00})>0\} }/\tilde p^e_i(c_{00})]=o(N_v^2)$.
	
	Note that $\text{Var}[1-(1-p)^{\tilde d_i}]=o(1)$ and $I_{\{\tilde{d}_i>0\} }\xrightarrow{P} 1$. Similar to the proof of Condition 1 satisfied in Proposition 2, we have
	\begin{align*}
	\lim_{N_v\rightarrow\infty} \mathbb E\Big[	\dfrac{p}{1-(1-p)^{\tilde d_i} }I_{\{\tilde{d}_i>0\} }\Big] =\lim_{N_v\rightarrow\infty}\dfrac{p}{\mathbb E[ 1-(1-p)^{\tilde d_i}]}.
	\end{align*}
	Since $\mathbb E[ 1-(1-p)^{\tilde d_i}]=\Theta(1-(1-p)^{d_i})$, we have $\mathbb E[I_{\{\tilde p^e_i(c_{11})>0\} }/\tilde p_i^e(c_{11})]=\mathcal O(1/ p_i^e(c_{11}))$ and $\mathbb E[I_{\{\tilde p^e_i(c_{01})>0\} }/\tilde p_i^e(c_{01})]=\mathcal O(1/ p_i^e(c_{01}))$. By Proposition 2, we obtain $\mathbb E[ \sum_{i=1}^{N_v}I_{\{\tilde p^e_i(c_{11})>0\} }/\tilde p^e_i(c_{11})]=o(N_v^2)$ and $\mathbb E[ \sum_{i=1}^{N_v}I_{\{\tilde p^e_i(c_{10})>0\} }/\tilde p^e_i(c_{10})]=o(N_v^2)$. 
	
	(iii) For all $c_k$, $\mathbb E[\sum_{i=1}^{N_v}\sum_{j\neq i}^{N_v}I_{\{\tilde p^e_i(c_k)>0\} }I_{\{\tilde p^e_j(c_k)>0\} }| \tilde p^e_{ij}(c_k)/(\tilde p^e_i(c_k)\tilde p^e_j(c_k))-1 |]=o(N_v^2)$. 
	
	Define a pairwise dependency indicator $\tilde g_{ij}$ in the observed network such that if $\tilde g_{ij}=0$, then $f({\bm Z}, {\tilde{ \bm A}}_{\cdot i})\independent f({\bm Z}, {\tilde{ \bm A}}_{\cdot j}  )$, else let $\tilde g_{ij}=1$. Then, we have $\tilde p^e_{ij}(c_k)=\tilde p^e_i(c_k)\tilde p^e_j(c_k)$ if $\tilde g_{ij}=0$. Thus, we obtain  
	
	\begin{align*}
	&\ \mathbb E[\sum_{i=1}^{N_v}\sum_{j\neq i}^{N_v}\tilde g_{ij}I_{\{\tilde p^e_i(c_k)>0\} }I_{\{\tilde p^e_j(c_k)>0\} }| \tilde p^e_{ij}(c_k)/(\tilde p^e_i(c_k)\tilde p^e_j(c_k))-1 |]\\
	=&\ \mathbb E[\sum_{i=1}^{N_v}\sum_{j\neq i}^{N_v} I_{\{\tilde p^e_i(c_k)>0\} }I_{\{\tilde p^e_j(c_k)>0\} }| \tilde p^e_{ij}(c_k)/(\tilde p^e_i(c_k)\tilde p^e_j(c_k))-1 |].
	\end{align*}

	By the triangle inequality, we have
	\begin{align*} 
	&\sum_{i=1}^{N_v}\sum_{j\neq i}^{N_v}\tilde g_{ij}I_{\{\tilde p^e_i(c_k)>0\} }I_{\{\tilde p^e_j(c_k)>0\} }| \tilde p^e_{ij}(c_k)/(\tilde p^e_i(c_k)\tilde p^e_j(c_k))-1 |\\
	\leq & \sum_{i=1}^{N_v}\sum_{j\neq i}^{N_v}\tilde g_{ij}I_{\{\tilde p^e_i(c_k)>0\} }I_{\{\tilde p^e_j(c_k)>0\} } \tilde p^e_{ij}(c_k)/(\tilde p^e_i(c_k)\tilde p^e_j(c_k))+\sum_{i=1}^{N_v}\sum_{j\neq i}^{N_v}\tilde g_{ij}I_{\{\tilde p^e_i(c_k)>0\} }I_{\{\tilde p^e_j(c_k)>0\} }.
	\end{align*}
	Note that $\sum_{i=1}^{N_v}\sum_{j\neq i}^{N_v}\tilde g_{ij}I_{\{\tilde p^e_i(c_k)>0\} }I_{\{\tilde p^e_j(c_k)>0\} }\leq \sum_{i=1}^{N_v}\sum_{j\neq i}^{N_v}\tilde g_{ij} $. Thus, it suffices to show $\mathbb E[\sum_{i=1}^{N_v}\sum_{j\neq i}^{N_v}\tilde g_{ij}I_{\{\tilde p^e_i(c_k)>0\} }I_{\{\tilde p^e_j(c_k)>0\} }\tilde p^e_{ij}(c_k)/(\tilde p^e_i(c_k)\tilde p^e_j(c_k))]=o(N_v^2)$ and $\mathbb E[\sum_{i=1}^{N_v}\sum_{j\neq i}^{N_v}\tilde g_{ij}]=o(N_v^2)$. 
	
	First, we prove $\mathbb E[\sum_{i=1}^{N_v}\sum_{j\neq i}^{N_v}\tilde g_{ij}]=o(N_v^2)$. Let $\tilde C_{ij}$ denote the number of common neighbors between vertex $i$ and $j$ in the observed network. Note that $\tilde g_{ij}=1$ if $\tilde C_{ij}\geq 1$ or $\tilde A_{i,j}=1$. Thus, it suffices to show $\mathbb E[\sum_{i=1}^{N_v}\sum_{j\neq i}^{N_v}I_{\{\tilde C_{ij}=0\}}]\sim N_v^2$ and $\mathbb E[\sum_{i=1}^{N_v}\sum_{j\neq i}^{N_v}I_{\{\tilde A_{i,j}=0\}}]\sim N_v^2$. A direct computation yields to
	\begin{align*}
	\mathbb E\Big[\sum_{i=1}^{N_v}\sum_{j\neq i}^{N_v}I_{\{\tilde A_{i,j}=0\}}\Big]&=(1-\alpha)\sum_{i=1}^{N_v}(N_v-1-d_i) +\beta\sum_{i=1}^{N_v}d_i\sim N_v^2.
	\end{align*} 
	Note that $\tilde C_{ij}=\sum_{k=1}^{N_v}\tilde A_{k,i}\tilde A_{k,j}\sim\text{Binomial}(C_{ij},(1-\beta)^2)+\text{Binomial}(d_i+d_j-2C_{ij},\alpha(1-\beta))+\text{Binomial}(N_v-d_i-d_j+C_{ij}-2,\alpha^2)$, and the three binomial random variables are independent. Then, we have
	\begin{align*}
	\mathbb E\Big[\sum_{i=1}^{N_v}\sum_{j\neq i}^{N_v}I_{\{\tilde C_{ij}=0\}}\Big]&=\sum_{i=1}^{N_v}\sum_{j\neq i}^{N_v}[1-(1-\beta)^2]^{C_{ij}}[1-\alpha(1-\beta)]^{d_i+d_j-2C_{ij}}(1-\alpha^2)^{ N_v-d_i-d_j+C_{ij}-2}.
	\end{align*} 
	Since $\sum_{i=1}^{N_v}\sum_{j\neq i}^{N_v}I_{\{  C_{ij}=0\}}\sim N_v^2$, we obtain $\mathbb E[\sum_{i=1}^{N_v}\sum_{j\neq i}^{N_v}I_{\{\tilde C_{ij}=0\}}]\sim N_v^2$.
	
	Next, we show $\mathbb E[\sum_{i=1}^{N_v}\sum_{j\neq i}^{N_v}\tilde g_{ij}I_{\{\tilde p^e_i(c_k)>0\} }I_{\{\tilde p^e_j(c_k)>0\} }\tilde p^e_{ij}(c_k)/(\tilde p^e_i(c_k)\tilde p^e_j(c_k))]=o(N_v^2)$. By Young's inequality and H\"older's inequality, we have
	\begin{align*}
	&\mathbb E\Big[\sum_{i=1}^{N_v}\sum_{j\neq i}^{N_v}\frac{\tilde g_{ij}I_{\{\tilde p^e_i(c_k)>0\} }I_{\{\tilde p^e_j(c_k)>0\} }\tilde p^e_{ij}(c_k)}{\tilde p^e_i(c_k)\tilde p^e_j(c_k)}\Big]\\
	\leq&\ \Bigg( \sum_{i=1}^{N_v}\sum_{j\neq i}^{N_v} \mathbb E(\tilde g_{ij})\Bigg)^{1/2}\Bigg( \sum_{i=1}^{N_v}\sum_{j\neq i}^{N_v}\mathbb E \Big[ \frac{I_{\{\tilde p^e_i(c_k)>0\} }I_{\{\tilde p^e_j(c_k)>0\} }\tilde p^e_{ij}(c_k)}{\tilde p^e_i(c_k)\tilde p^e_j(c_k)}\Big]^2\Bigg)^{1/2}.
	\end{align*}
	For the level $c_{00}$, we have 
	\begin{align*}
	\sum_{i=1}^{N_v}\sum_{j\neq i}^{N_v}\mathbb E  \Big[ \frac{ I_{\{\tilde p^e_i(c_{00})>0\} }I_{\{\tilde p^e_j(c_{00})>0\} }\tilde p^e_{ij}(c_{00})}{\tilde p^e_i(c_{00})\tilde p^e_j(c_{00})}\Big]^2\leq  \sum_{i=1}^{N_v}\sum_{j\neq i}^{N_v}\mathbb E  \Big[ \frac{ 1}{\tilde p^e_i(c_{00})\tilde p^e_j(c_{00})}\Big]^2 
	\end{align*}
	and $\mathbb E [(1/\tilde p^e_i(c_{00})\tilde p^e_j(c_{00}) )^2 ]=\mathcal O((1/p^e_i(c_{00})p^e_j(c_{00}) )^2)$. Thus, we obtain 
	\begin{align}\label{eqb.1}
	\mathbb E[\sum_{i=1}^{N_v}\sum_{j\neq i}^{N_v} I_{\{\tilde p^e_i(c_{00} )>0\} }I_{\{\tilde p^e_j(c_{00})>0\} }\tilde p^e_{ij}(c_{00})/(\tilde p^e_i(c_{00})\tilde p^e_j(c_{00}))]=o(N_v^2).
	\end{align}
	Similarly, we can show (\ref{eqb.1}) for other exposure levels. These complete the proof.

	\subsection{Proof of Proposition 5}
	
	The proof of Proposition 5 is the same as that of Proposition 4.
	
	\section{Proofs of theorems and corollary for noisy networks}
	
	First, we show the following lemmas that will be used in the proofs of theorems for noisy networks.
	
	\begin{lemma}\label{l1}
		Let $X_1,X_2,\cdots $ and $Y_1,Y_2,\cdots $ be two nonnegative random variable sequences. Assume $X_nI_{\{Y_n>0\}}/Y_n$ is bounded and $\mathbb E(Y_n)\neq 0$ for all $n$, $\text{Var}(X_n)=o(1)$, $\text{Var}(Y_n)=o(1)$, and $I_{\{Y_n>0\}}\xrightarrow{P}1$ as $n\rightarrow\infty$. Then, we have $ \mathbb E[X_nI_{\{Y_n>0\}}/Y_n]=  \mathbb E (X_n)/\mathbb E (Y_n)+o(1)$.
	\end{lemma}
	\begin{proof}
		Since $\text{Var}(X_n)=o(1)$, we have $X_n\xrightarrow{L_2} \lim_{n\rightarrow\infty}\mathbb E(X_n)$. This implies $X_n\xrightarrow{P} \lim_{n\rightarrow\infty}\mathbb E (X_n)$. Similarly, by continuous mapping theorem, we have $I_{\{Y_n>0\}}/Y_n\xrightarrow{P} \lim_{n\rightarrow\infty}1/\mathbb E (Y_n)$. Thus, we obtain $X_nI_{\{Y_n>0\}}/Y_n\xrightarrow{P} \lim_{n\rightarrow\infty}\mathbb E (X_n)/\mathbb E (Y_n)$. Since $X_nI_{\{Y_n>0\}}/Y_n$ is bounded, we show $\lim_{n\rightarrow\infty}\mathbb E[X_nI_{\{Y_n>0\}}/Y_n]= \lim_{n\rightarrow\infty}\mathbb E (X_n)/\mathbb E (Y_n)$. This completes the proof.
	\end{proof}
	
	\begin{lemma}\label{l2}
		Let $T_1,T_2,\cdots $, $X_1,X_2,\cdots $, and $Y_1,Y_2,\cdots $ be three random variable sequences. Assume $X_n=a+\mathcal O_p(g(n))$, $Y_n=b+\mathcal O_p(h(n))$ as $n\rightarrow\infty$, where $g(n)=o(a)$ and $h(n)=o(b)$. Let $f_1(x,y,t),f_2(x,y,t),\cdots$ be a sequence of function. And $f_n(a,b,T_n)\xrightarrow{P}c$ as $n\rightarrow\infty$. In addition, we assume $f_n(X_n,Y_n,T_n)-f_n(a,b,T_n)=o(1)$ when $|X_n-a|\leq c_1 g(n)$ and $|Y_n-b|\leq c_2 h(n)$, where $\mathbb P(|X_n-a|> c_1g(n) )=o(1)$ and $\mathbb P(|Y_n-b|> c_2 h(n) )=o(1), c_1=\omega(1),c_2=\omega(1)$. Then, we have $f_n(X_n,Y_n,T_n)\xrightarrow{P}c$ as $n\rightarrow\infty$. 
	\end{lemma}
	\begin{proof}
		By triangle inequality, we have 
		\begin{align*}
		|f_n(X_n,Y_n,T_n)-c|\leq |f_n(X_n,Y_n,T_n)-f(a,b,T_n)|+| f_n(a,b,T_n)-c|.
		\end{align*}	
		Then, for any $\varepsilon_n>0$, we have
		\begin{align*}
		&\ \mathbb P(|f_n(X_n,Y_n,T_n)-c|>\varepsilon_n)\\
		\leq&\ \mathbb P(|f_n(X_n,Y_n,T_n)-f_n(a,b,T_n)|+| f_n(a,b,T_n)-c|>\varepsilon_n)\\
		\leq&\ \mathbb P(|f_n(X_n,Y_n,T_n)-f_n(a,b,T_n)|>\varepsilon_n/2 )+\mathbb P(| f_n(a,b,T_n)-c|>\varepsilon_n/2 )
		\end{align*}
		where the last step follows by the pigeonhole principle and the sub-additivity of the probability measure. Note that $\lim_{n\rightarrow\infty} \mathbb P(|f_n(X_n,Y_n,T_n)-f_n(a,b,T_n)|>\varepsilon_n/2 ) =0$ and $\lim_{n\rightarrow\infty} \mathbb P(| f_n(a,b,T_n)-c|>\varepsilon_n/2 ) =0$. Thus, we obtain 
		\begin{align*}
		\lim_{n\rightarrow\infty}\mathbb P(|f(_nX_n,Y_n,T_n)-c|>\varepsilon)=0.
		\end{align*}
	\end{proof}

	\begin{lemma}\label{l4}
		Let $T_1,T_2,\cdots $, $X_1,X_2,\cdots $, and $Y_1,Y_2,\cdots $ be three random variable sequences. Assume $X_n=a+\mathcal O_p(g(n))$, $Y_n=b+\mathcal O_p(h(n))$ as $n\rightarrow\infty$, where $g(n)=o(a)$ and $h(n)=o(b)$. Let $f_1(x,y,t),f_2(x,y,t),\cdots$ be a sequence of bounded function. And $f_n(X_n,Y_n,T_n)-f_n(a,b,T_n)=o(1)$ when $|X_n-a|\leq c_1 g(n)$ and $|Y_n-b|\leq c_2 h(n)$, where $\mathbb P(|X_n-a|> c_1g(n) )=o(1)$ and $\mathbb P(|Y_n-b|> c_2 h(n) )=o(1), c_1=\omega(1),c_2=\omega(1)$. Then, we have $\mathbb E f_n(X_n,Y_n,T_n)-\mathbb E f_n(a,b,T_n)\rightarrow0$ as $n\rightarrow\infty$. 
	\end{lemma}
	
	\begin{proof}
		By the definition of the expectation, we have
		\begin{align*}
		&\mathbb E f_n(X_n,Y_n,T_n)-\mathbb E f_n(a,b,T_n) \\
		= \ &\mathbb E\Big[ \big(f_n(X_n,Y_n,T_n) -f_n(a,b,T_n)  ) \cdot I_{\{|X_n-a|>c_1g(n) \text{ or } |Y_n-b|> c_2 h(n) \}  } \Big] \\ 
		+\ &  \mathbb E\Big[ \big(f_n(X_n,Y_n,T_n) -f_n(a,b,T_n)  ) \cdot I_{\{|X_n-a|\leq c_1 g(n)  \text{ and } |Y_n-b|\leq  c_2 h(n)  \}  } \Big] 
		\end{align*}
		Since $f_n$ is bounded, we have 
		\begin{align*}
		\mathbb E f_n(X_n,Y_n,T_n)-\mathbb E f_n(a,b,T_n)\rightarrow 0 
		\end{align*}
		as $n\rightarrow\infty$. 
	\end{proof}

	\begin{lemma}\label{ll1}
		Let $T_1,T_2,\cdots $, $X_1,X_2,\cdots $ and $Y_1,Y_2,\cdots $ be three random variable sequences, and $X_n$ and $Y_n$ depend  on $T_n$ for all $n$. Assume $X_n$ is bounded, $\mathbb E(Y_n)\neq 0$ for all $n$, $\mathbb E(Y_n)= \Theta(1)$,  $\text{Var}(Y_n)=o(1)$ as $n\rightarrow\infty$. Then, we have 
		\begin{align*}
		\mathbb E\left(\frac{X_n}{Y_n} I_{ \{Y_n=\Theta(1)\cap Y_n\neq 0 \} }\right) = \frac{\mathbb EX_n}{\mathbb EY_n} + \mathcal O\left(\max\left\{ (\text{Var} Y_n)^{1/3},1- \mathbb P(Y_n=\Theta(1)\cap Y_n\neq 0)\right\} \right).
		\end{align*}
	\end{lemma}
	
	\begin{proof}
		By Jensen's inequality, we have
		\begin{align} \label{eq:c1}
		&\ \left|\mathbb E\left(\frac{X_n}{Y_n} I_{ \{Y_n=\Theta(1)\cap Y_n\neq 0 \} }\right)  - \frac{\mathbb EX_n}{\mathbb EY_n}   \right|\\
		=&\ \mathcal O\Bigg(\max\bigg\{  \mathbb E  \left(\frac{X_n |\mathbb EY_n-Y_n| }{Y_n} \cdot I_{ \{Y_n=\Theta(1) \cap Y_n\neq 0  \} } \right) , 1- \mathbb P(Y_n=\Theta(1)\cap Y_n\neq 0)\bigg\}\Bigg).
		\end{align}
		Then by additivity of expectation, for $0<\delta_n<1$, $\mathbb E  \left(\frac{X_n |\mathbb EY_n-Y_n| }{Y_n}\cdot I_{ \{Y_n=\Theta(1)\cap Y_n\neq 0  \} }\right)$ in (\ref{eq:c1}) equals
		\begin{align}\label{eq:c2}
		\begin{split}
		&\ \mathbb E  \left(\frac{X_n |\mathbb EY_n-Y_n| }{Y_n} \cdot I_{ \{Y_n=\Theta(1)\cap Y_n\neq 0  \} }\cdot I_{|Y_n-\mathbb EY_n|\geq \delta_n \mathbb EY_n }  \right)  \\
		+ &\ \mathbb E  \left(\frac{X_n |\mathbb EY_n-Y_n| }{Y_n} \cdot I_{ \{Y_n=\Theta(1)\cap Y_n\neq 0  \} }\cdot I_{|Y_n-\mathbb EY_n|< \delta_n \mathbb EY_n }  \right)    .
		\end{split}
		\end{align}
		Next, we find the upper bounds of two terms in (\ref{eq:c2}). For the first term, by Chebyshev's inequality, we obtain
		\begin{align}
		\mathbb E  \left(\frac{X_n |\mathbb EY_n-Y_n| }{Y_n} \cdot I_{ \{Y_n=\Theta(1)\cap Y_n\neq 0  \} }\cdot I_{|Y_n-\mathbb EY_n|\geq \delta_n \mathbb EY_n }  \right) =\mathcal O\left( \frac{\text{Var} Y_n }{\delta_n^2   } \right).
		\end{align}
		For the second term, when $|Y_n-\mathbb E Y_n|<\delta_n\mathbb E Y_n$, $Y_n> (1-\delta_n)\mathbb E Y_n$. So, we obtain
		\begin{align*}
		\mathbb E  \left(\frac{X_n |\mathbb EY_n-Y_n| }{Y_n} \cdot I_{ \{Y_n=\Theta(1)\cap Y_n\neq 0  \} }\cdot I_{|Y_n-\mathbb EY_n|< \delta_n \mathbb EY_n }  \right) =\mathcal O\left( \frac{\delta_n}{(1-\delta_n) }\right).
		\end{align*}
		Letting $\delta_n=(\text{Var} Y_n)^{1/3}$ completes the proof.
	\end{proof}

	\begin{lemma}\label{l5}
		Define $\hat \theta_n=\theta_n(a,b,T_n)$, $\tilde \theta_n=\theta_n(X_n,Y_n,T_n)$, where $X_n$, $Y_n$, $T_n$ are random variables, and $a,b$ are constants. Assume $\frac{\hat \theta_n - \tilde \theta_n}{\sqrt{\text{Var}(\tilde \theta_n) }}\xrightarrow{P} 0$ as $n\rightarrow\infty$. Then, we have 
		\begin{itemize}
			\item [(i)] $\text{Var}(\hat \theta_n) \sim \text{Var}(\tilde \theta_n)$ as $n\rightarrow\infty$.
			\item [(ii)] $\frac{\hat \theta_n - \mathbb E(\hat \theta_n )}{\sqrt{\text{Var}(\hat \theta_n) }}$ and $\frac{\tilde \theta_n - \mathbb E(\tilde  \theta_n )}{\sqrt{\text{Var}(\tilde \theta_n) }}$ have the same asymptotic distribution.
		\end{itemize}
	\end{lemma}
	\begin{proof}
		By Slutsky's theorem, we obtain that $\frac{\hat \theta_n - \mathbb E(\hat \theta_n )}{\sqrt{\text{Var}(\tilde \theta_n) }}$ and $\frac{\tilde \theta_n - \mathbb E(\hat  \theta_n )}{\sqrt{\text{Var}(\tilde \theta_n) }}$ have the same asymptotic distribution. 
		Next, we show statements (i) and (ii) in this lemma. Let $W$ denote the limiting distribution of $\frac{\hat \theta_n - \mathbb E(\hat \theta_n )}{\sqrt{\text{Var}(\hat \theta_n) }}$. Note that $\mathbb E(W)=0$ and $\text{Var}(W)=1$.
		Assume that $\sqrt{\text{Var}(\hat \theta_n)}/\sqrt{\text{Var}(\tilde \theta_n)}\rightarrow c_1$ and $ \big(\mathbb E(\hat\theta_n)-\mathbb E(\tilde\theta_n) \big)/\sqrt{\text{Var}(\tilde \theta_n)}\rightarrow c_2$.
		By Slutsky's theorem, we have $ \frac{\tilde \theta_n - \mathbb E(\tilde  \theta_n )}{\sqrt{\text{Var}(\tilde \theta_n) }}  \xrightarrow{d}  c_1 W+ c_2$.
		Since $\mathbb E\Big(\frac{\tilde \theta_n - \mathbb E(\tilde  \theta_n )}{\sqrt{\text{Var}(\tilde \theta_n) }}\Big)=0$ and $\text{Var}\Big(\frac{\tilde \theta_n - \mathbb E(\tilde  \theta_n )}{\sqrt{\text{Var}(\tilde \theta_n) }}\Big)=1$, we obtain $c_1=1$ and $c_2=0$.
		Thus, we have $\text{Var}(\hat \theta_n) \sim \text{Var}(\tilde \theta_n)$ and $ \frac{\tilde \theta_n - \mathbb E(\tilde  \theta_n )}{\sqrt{\text{Var}(\tilde \theta_n) }}  \xrightarrow{d}  W$.
	\end{proof}

	Mimicking the proofs of Lemma 1, Proposition 1, Theorem 2 and Theorem 3 in \cite{chang2020estimation}, we can show the following lemma.
	\begin{lemma}\label{ll2}
		Under the assumptions 1-5, it holds that  $\hat\alpha=\alpha+\mathcal O_p(N_v^{-3/2}\bar d^{-1/2})$, $\hat\beta=\beta+\mathcal O_p(N_v^{-1/2}\bar d^{-1/2})$  and $(N_v^{3/2}\bar d^{1/2}\hat\alpha-N_v^{3/2}\bar d^{1/2} \alpha,N_v^{ 1/2}\bar d^{1/2}\hat\beta- N_v^{ 1/2}\bar d^{1/2} \beta)$ has asymptotic normal distribution with mean $(0,0)$. 
	\end{lemma}
	
	We summarize tails bounds for Normal distribution and Poisson distribution in the following lemma. 
	\begin{lemma}\label{ll3}
		Let X be a random variable. Define $\mu=\mathbb EX$.
		\begin{itemize}
			\item [(i)] If $X\sim \text{Pois}(\lambda)$, $\mathbb P(X\geq x)\leq \dfrac{(e\lambda)^xe^{-\lambda}}{x^x}$, for $x>\lambda$, and $\mathbb P(X\leq x)\leq \dfrac{(e\lambda)^xe^{-\lambda}}{x^x}$, for $x<\lambda$.
			\item [(ii)] If $X\sim \text{Normal}(\mu,\sigma^2)$, $\mathbb P(|X-\mu|\geq \sigma x)\leq 2\ \dfrac{e^{-x^2/2} }{x\sqrt{2\pi} }$, for $x>0$.
		\end{itemize}
	\end{lemma}
	
	\subsection{Proof of Theorem 1}
	
	In the noisy network, Aronow and Samii estimators for $\overline{y}(c_k)$ become 
	\begin{align*}
	\tilde{\overline{y}}_{A\&S}(c_{11})=&\ \frac{1}{N_v} \sum_{i=1}^{N_v}I_{ \{\tilde d_i>0 \} }  I_{\{f({\bm Z},\tilde {\bm A}_{\cdot i})=c_{11}\}} \frac{1}{p\big[1-(1-p)^{\tilde d_i}\big] } \Big[I_{\{f({\bm Z}, {\bm A}_{\cdot i})=c_{11}\}}y_i(c_{11}) \\ &\ +I_{\{f({\bm Z}, {\bm A}_{\cdot i})=c_{10}\}} y_i(c_{10})\Big] ,\\
	\tilde{\overline{y}}_{A\&S}(c_{10})=&\ \frac{1}{N_v} \sum_{i=1}^{N_v}I_{\{f({\bm Z},\tilde {\bm A}_{\cdot i})=c_{10}\}}\frac{1}{p(1-p)^{\tilde d_i}} \Big[I_{\{f({\bm Z}, {\bm A}_{\cdot i})=c_{11}\}} y_i(c_{11})+I_{\{f({\bm Z}, {\bm A}_{\cdot i})=c_{10}\}} y_i(c_{10}) \Big] ,\\
	\tilde{\overline{y}}_{A\&S}(c_{01})=&\ \frac{1}{N_v} \sum_{i=1}^{N_v}I_{ \{\tilde d_i>0 \} }  I_{\{f({\bm Z},\tilde {\bm A}_{\cdot i})=c_{01}\}} \frac{1}{(1-p)\big[1-(1-p)^{\tilde d_i}\big] } \Big[I_{\{f({\bm Z}, {\bm A}_{\cdot i})=c_{01}\}}y_i(c_{01}) \\ &\ +I_{\{f({\bm Z}, {\bm A}_{\cdot i})=c_{00}\}} y_i(c_{00})\Big] ,\\
	\tilde{\overline{y}}_{A\&S}(c_{00})=&\ \frac{1}{N_v} \sum_{i=1}^{N_v}I_{\{f({\bm Z},\tilde {\bm A}_{\cdot i})=c_{00}\}}\frac{1}{(1-p)^{\tilde d_i+1}} \Big[I_{\{f({\bm Z}, {\bm A}_{\cdot i})=c_{01}\}} y_i(c_{01})+I_{\{f({\bm Z}, {\bm A}_{\cdot i})=c_{00}\}} y_i(c_{00}) \Big].
	\end{align*}
	The conditional biases for these estimators are
	\begin{align*}
	\text{Bias}\Big[ \tilde{\overline{y}}_{A\&S}(c_{11})|\tilde {\bm A}\Big] = & \ -\frac{1}{N_v}\sum_{i=1}^{N_v}I_{ \{\tilde d_i>0 \} } \frac{(1-p)^{d_i}\big[1-(1-p)^{\tilde d_i-\check d_i}\big]}{1-(1-p)^{\tilde d_i}}\  \uptau_i(c_{11},c_{10})\\
	& \  -\frac{1}{N_v}\sum_{i=1}^{N_v}I_{ \{\tilde d_i=0 \} }\ y_i(c_{11}),\\
	\text{Bias}\Big[ \tilde{\overline{y}}_{A\&S}(c_{10})|\tilde {\bm A}\Big]= &\ \frac{1}{N_v}\sum_{i=1}^{N_v}\big[1-(1-p)^{d_i-\check d_i}\big] \uptau_i(c_{11},c_{10}),\\
	\text{Bias}\Big[ \tilde{\overline{y}}_{A\&S}(c_{01})|\tilde {\bm A}\Big]= & \ -\frac{1}{N_v}\sum_{i=1}^{N_v}I_{ \{\tilde d_i>0 \} } \frac{(1-p)^{d_i}\big[1-(1-p)^{\tilde d_i-\check d_i}\big]}{1-(1-p)^{\tilde d_i}}\ \uptau_i(c_{01},c_{00})\\
	& \  -\frac{1}{N_v}\sum_{i=1}^{N_v}I_{ \{\tilde d_i=0 \} }\ y_i(c_{01}),\\
	\text{Bias}\Big[ \tilde{\overline{y}}_{A\&S}(c_{00})|\tilde {\bm A}\Big]=  &\ \frac{1}{N_v}\sum_{i=1}^{N_v}\big[1-(1-p)^{d_i-\check d_i}\big] \uptau_i(c_{01},c_{00}). 
	\end{align*}
	
	Based on the noisy network model, we have
	\begin{align*}
	\mathbb{E}[(1-p)^{d_i-\check d_i}]&=(1-\beta p)^{d_i}, \\
	\text{Var}[(1-p)^{  d_i-\check d_i}]&=\big[1-\beta p(2-p)\big]^{ d_i}-(1-\beta p)^{2 d_i},\\
	\mathbb{E}[(1-p)^{\tilde d_i-\check d_i}]&=(1-\alpha p)^{N_v-1-d_i},\\
	\text{Var}[(1-p)^{\tilde d_i-\check d_i}]&=\big[1-\alpha p(2-p)\big]^{N_v-1-d_i}-(1-\alpha p)^{2(N_v-1-d_i)},\\
	\mathbb{E}[(1-p)^{\tilde d_i}]=&\ (1-\alpha p)^{N_v-1-d_i}\big[1-(1-\beta)p\big]^{d_i},\\
	\text{Var}[(1-p)^{\tilde d_i}]=&\ \big[1-\alpha p(2-p)\big]^{N_v-1-d_i}\big[1-(1-\beta)p(2-p)\big]^{d_i},\\
	&-(1-\alpha p)^{2(N_v-1-d_i)}\big[1-(1-\beta)p\big]^{2d_i}. 
	\end{align*}
	For exposure levels $c_{10}$ and $c_{00}$, by taking the expectation with respect to $\tilde {\bm A}$, we obtain 
	\begin{align*}
	\text{Bias}\Big[ \tilde{\overline{y}}_{A\&S}(c_{10})\Big]=  &\ \frac{1}{N_v} \sum_{i=1}^{N_v}\big[1-(1-\beta p)^{d_i}\big] \ \uptau_i(c_{11},c_{10}),\\
	\text{Bias}\Big[ \tilde{\overline{y}}_{A\&S}(c_{00})\Big]=  &\ \frac{1}{N_v} \sum_{i=1}^{N_v}\big[1-(1-\beta p)^{d_i}\big] \ \uptau_i(c_{01},c_{00}).
	\end{align*}
	Recall that $p=o(1)$. Then, we have  $\text{Var}[(1-p)^{d_i-\check d_i}]=o(1)$, $\text{Var}[(1-p)^{\tilde d_i}]=o(1)$  and $I_{\{\tilde d_i>0\}}\xrightarrow{P}1$. For exposure levels $c_{11}$ and $c_{01}$, by Lemma \ref{l1}, we have
	\begin{align*}
	\text{Bias}\Big[ \tilde{\overline{y}}_{A\&S}(c_{11})\Big] =& \ -  \frac{1}{N_v} \sum_{i=1}^{N_v} \frac{(1-p)^{d_i}\big[1-(1-\alpha p)^{N_v-1-d_i}\big]}{1-(1-\alpha p)^{N_v-1-d_i}(1- (1-\beta)p  )^{d_i}}\ \uptau_i(c_{11},c_{10})+  o(1),\\
	\text{Bias}\Big[ \tilde{\overline{y}}_{A\&S}(c_{01})\Big] =& \ -  \frac{1}{N_v} \sum_{i=1}^{N_v} \frac{(1-p)^{d_i}\big[1-(1-\alpha p)^{N_v-1-d_i}\big]}{1-(1-\alpha p)^{N_v-1-d_i}(1- (1-\beta)p  )^{d_i}}\ \uptau_i(c_{01},c_{00}) +  o(1).
	\end{align*}

	\subsection{Proof of Theorem 2}
	
	By the law of total variance, we have $\text{Var}[ \tilde{\overline{y}}_{A\&S}(c_{k})]=\text{Var}[\mathbb E( \tilde{\overline{y}}_{A\&S}(c_{k})|\tilde {\bm A} )]+\mathbb E[\text{Var}( \tilde{\overline{y}}_{A\&S}(c_{k})|\tilde {\bm A} )]$. Thus, it suffices to show $\text{Var}[\mathbb E( \tilde{\overline{y}}_{A\&S}(c_{k})|\tilde {\bm A} )]=o(1)$ and  $\mathbb E[\text{Var}( \tilde{\overline{y}}_{A\&S}(c_{k})|\tilde {\bm A} )]=o(1)$.
	
	(i) $\text{Var}[\mathbb E( \tilde{\overline{y}}_{A\&S}(c_{k})|\tilde {\bm A} )]=o(1)$.
	
	Note that $\text{Var}[\mathbb E( \tilde{\overline{y}}_{A\&S}(c_{k})|\tilde {\bm A} )]=\text{Var}[\text{Bias}( \tilde{\overline{y}}_{A\&S}(c_{k})|\tilde {\bm A} )]$. For the exposure level $c_{00}$, we have
	\begin{align*}
	\text{Var}[\mathbb E( \tilde{\overline{y}}_{A\&S}(c_{00})|\tilde {\bm A} )]
	=  &\ \frac{1}{N_v^2}\sum_{i=1}^{N_v}\text{Var}\big[1-(1-p)^{d_i-\check d_i}\big] \uptau_i^2(c_{01},c_{00}) \\
	&+ \frac{1}{N_v^2}\sum_{i=1}^{N_v}\sum_{j\neq i} \text{Cov}\big((1-p)^{d_i-\check d_i},(1-p)^{d_j-\check d_j}\big )\uptau_i (c_{01},c_{00})\uptau_j(c_{01},c_{00}).
	\end{align*}
	Note that $\text{Var}[1-(1-p)^{d_i-\check d_i}]\leq 1$ and $\text{Cov}((1-p)^{d_i-\check d_i},(1-p)^{d_j-\check d_j} )=0$ if $A_{i,j}=0$. Thus, we have $\text{Var}[\mathbb E( \tilde{\overline{y}}_{A\&S}(c_{00})|\tilde {\bm A} )]=o(1)$. Similarly, we can show $\text{Var}[\mathbb E( \tilde{\overline{y}}_{A\&S}(c_{k})|\tilde {\bm A} )]=o(1)$ for other exposure levels.
	
	(ii) $\mathbb E[\text{Var}( \tilde{\overline{y}}_{A\&S}(c_{k})|\tilde {\bm A} )]=o(1)$.
	
	For the exposure level $c_{00}$, we have
	\begin{align*}\label{eqc.1}
	&\ \text{Var} \Big[ \tilde{\overline{y}}_{A\&S}(c_{00})\Big|\tilde {\bm A} \Big] \\
	=&\ \frac{1}{N_v^2}
	\Bigg\{ \sum_{i=1}^{N_v}\Big\{ \sum_{k'\in\{00,01\} } \check {p}_i^e(c_{00},c_{k'}) \Big[1-\check {p}_i^e(c_{00} ,c_{k'})\Big]\Big[\frac{y_i(c_{k'})}{\tilde p_i^e(c_{00} )} \Big]^2- 2\check {p}_i^e(c_{00} ,c_{00})\check {p}_i^e(c_{00} ,c_{01})\frac{y_i(c_{00})}{\tilde p_i^e(c_{00} )}\frac{y_i(c_{01})}{\tilde p_i^e(c_{00} )}\Big\} \\
	&\ +\sum_{i=1}^{N_v} \sum_{j\neq i }\sum_{k'\in\{00,01\}  } \sum_{k''\in\{00,01\}  } \big[\check {p}_{ij}^e(c_{00} ,c_{k'},c_{00},c_{k''})-\check {p}_{i}^e(c_{00},c_{k'})\check {p}_{j}^e(c_{00},c_{k''})\big]\frac{y_{i}(c_{k'})}{\tilde p_{i}^e(c_{00})}\frac{y_{j}(c_{k''})}{\tilde p_{j}^e(c_{00})}\Bigg \},\numberthis
	\end{align*}
	where
	\begin{align*} 
	\check {p}_{ij}^e(c_{k_1},c_{k_2} )&=\sum_{\bm z}p_{\bm z}I_{\{f({\bm z}, \tilde {\bm A}_{\cdot i})=c_{k_1}\}}I_{\{f({\bm z}, {\bm A}_{\cdot i})=c_{k_2}\}},\\
	\check {p}_{ij}^e(c_{k_1},c_{k_2},c_{k_3},c_{k_4})&=\sum_{\bm z}p_{\bm z} I_{\{f({\bm z},\tilde {\bm X}_{i})=c_{k_1}\}} I_{\{f({\bm z},{\bm x}_{i})=c_{k_2}\}} I_{\{f({\bm z},\tilde {\bm X}_{j})=c_{k_3}\}} I_{\{f({\bm z},{\bm x}_{j})=c_{k_4}\}}  
	\end{align*}
	for $k_1,\ k_2,\ k_3,\ k_4=1,2,\cdots, K$, $i,j=1,2,\cdots, N_v$.
	
	Note that $\sum_{k'\in\{00,01\} } \check {p}_i^e(c_{00},c_{k'})=\tilde p_i^e(c_{00})$ and $\sum_{k'\in\{00,01\}  } \sum_{k''\in\{00,01\}  }  \check {p}_{ij}^e(c_{00} ,c_{k'},c_{00},c_{k''})=\tilde p_{ij}^e(c_{00})$. Then, (\ref{eqc.1}) leads to 
	\begin{align*}
	\text{Var} \Big[ \tilde{\overline{y}}_{A\&S}(c_{00})\Big|\tilde {\bm A} \Big]\leq \frac{C_1}{N_v^2}    \sum_{i=1}^{N_v} \frac{1}{\tilde p_i^e(c_{00})}+\frac{C_2}{N_v^2} \sum_{i=1}^{N_v} \sum_{j\neq i }\frac{\tilde g_{ij}\tilde p_{ij}^e(c_{00}) }{\tilde p_i^e(c_{00})\tilde p_j^e(c_{00})},
	\end{align*}
	where $C_1$ and $C_2$ are positive constants. By Proposition 4 and (\ref{eqb.1}), we obtain $\mathbb E[\text{Var}( \tilde{\overline{y}}_{A\&S}(c_{00})|\tilde {\bm A} )]=o(1)$. Similarly, we can show $\mathbb E[\text{Var}( \tilde{\overline{y}}_{A\&S}(c_{k})|\tilde {\bm A} )]=o(1)$ for other exposure levels.
	
	\subsection{Proof of Theorem 3}
	
	The proof of Theorem 3 is same as that of Theorem 2.
	
	\subsection{Proof of Theorem 4}
	We first provide explicit expression for $\tilde y_{\text{MME},i}(c_{k})$.
	By the definition of $\tilde{\bm y}_i$ in (4.1), we have
	\begin{align*}
	\tilde y_i (c_{11})&=I_{\{f({\bm Z},\tilde {\bm A}_{\cdot i})=c_{11}\}}\Big[I_{\{f({\bm Z}, {\bm A}_{\cdot i})=c_{11}\}} y_i(c_{11})+I_{\{f({\bm Z}, {\bm A}_{\cdot i})=c_{10}\}}y_i(c_{10}) \Big] ,\\
	\tilde y_i (c_{10})&=I_{\{f({\bm Z},\tilde {\bm A}_{\cdot i})=c_{10}\}}\Big[I_{\{f({\bm Z}, {\bm A}_{\cdot i})=c_{11}\}} y_i(c_{11})+I_{\{f({\bm Z}, {\bm A}_{\cdot i})=c_{10}\}}y_i(c_{10}) \Big] ,\\
	\tilde y_i (c_{01})&=I_{\{f({\bm Z},\tilde {\bm A}_{\cdot i})=c_{01}\}}\Big[I_{\{f({\bm Z}, {\bm A}_{\cdot i})=c_{01}\}} y_i(c_{01})+I_{\{f({\bm Z}, {\bm A}_{\cdot i})=c_{00}\}}y_i(c_{00}) \Big], \\
	\tilde y_i (c_{00})&=I_{\{f({\bm Z},\tilde {\bm A}_{\cdot i})=c_{00}\}}\Big[I_{\{f({\bm Z}, {\bm A}_{\cdot i})=c_{01}\}} y_i(c_{01})+I_{\{f({\bm Z}, {\bm A}_{\cdot i})=c_{00}\}}y_i(c_{00}) \Big].
	\end{align*}
	By the definition of $\bm P(d_i,\alpha,\beta)$ in (4.2), we have $\bm P(d_i,\alpha,\beta) =\text{diag}(\bm S(d_i,\alpha,\beta),\bm Q(d_i,\alpha,\beta))$, where
	\begin{align*}
	\bm S(d_i,\alpha,\beta) &=\begin{bmatrix}
	\text{E}[ I_{\{f({\bm Z},\tilde {\bm A}_{\cdot i})=c_{11}\}}I_{\{f({\bm Z}, {\bm A}_{\cdot i})=c_{11}\}}]& \text{E}[ I_{\{f({\bm Z},\tilde {\bm A}_{\cdot i})=c_{11}\}}I_{\{f({\bm Z}, {\bm A}_{\cdot i})=c_{10}\}}]\\
	\text{E}[ I_{\{f({\bm Z},\tilde {\bm A}_{\cdot i})=c_{10}\}}I_{\{f({\bm Z}, {\bm A}_{\cdot i})=c_{11}\}}]& \text{E}[ I_{\{f({\bm Z},\tilde {\bm A}_{\cdot i})=c_{10}\}}I_{\{f({\bm Z}, {\bm A}_{\cdot i})=c_{10}\}}]  \\
	\end{bmatrix},\\
	\bm Q(d_i,\alpha,\beta) &=\begin{bmatrix}
	\text{E}[ I_{\{f({\bm Z},\tilde {\bm A}_{\cdot i})=c_{01}\}}I_{\{f({\bm Z}, {\bm A}_{\cdot i})=c_{01}\}}]& \text{E}[ I_{\{f({\bm Z},\tilde {\bm A}_{\cdot i})=c_{01}\}}I_{\{f({\bm Z}, {\bm A}_{\cdot i})=c_{00}\}}]\\
	\text{E}[ I_{\{f({\bm Z},\tilde {\bm A}_{\cdot i})=c_{00}\}}I_{\{f({\bm Z}, {\bm A}_{\cdot i})=c_{01}\}}]& \text{E}[ I_{\{f({\bm Z},\tilde {\bm A}_{\cdot i})=c_{00}\}}I_{\{f({\bm Z}, {\bm A}_{\cdot i})=c_{00}\}}]  \\
	\end{bmatrix},
	\end{align*}
	and
	\begin{align*}
	S_{11} (d_i,\alpha,\beta) =&\ p\Big\{ 1-(1-p)^{d_i} -(1-\alpha p)^{N_v-1-d_i} \cdot \big[  (1- (1-\beta)p )^{d_i}-(1-p)^{d_i} \big]  \Big \}, \\
	S_{12} (d_i,\alpha,\beta)=&\ p(1-p)^{d_i} \big[1-(1-\alpha p)^{N_v-1-d_i} \big], \\
	S_{21} (d_i,\alpha,\beta)=&\ p (1-\alpha p)^{N_v-1-d_i}  \big[  (1- (1-\beta)p )^{d_i}-(1-p)^{d_i} \big]  , \\
	S_{22} (d_i,\alpha,\beta) =&\ p(1-p)^{d_i} (1-\alpha p)^{N_v-1-d_i} ,\\
	\bm S (d_i,\alpha,\beta) =&\ \frac{1-p}{p}\cdot \bm Q(d_i,\alpha,\beta).
	\end{align*}
	Then, $\tilde{\bm y}_{\text{MME},i}$ in (4.3) can be written as
	\begin{align*} 
	\tilde y_{\text{MME},i }(c_{11})&= S_{11}^{-1}(\hat d_i,\hat\alpha,\hat\beta) \cdot \tilde y_i (c_{11})+S_{12}^{-1}(\hat d_i,\hat\alpha,\hat\beta) \cdot \tilde y_i (c_{10}) ,\\
	\tilde y_{\text{MME},i }(c_{10})&= S_{21}^{-1}(\hat d_i,\hat\alpha,\hat\beta) \cdot \tilde y_i (c_{11})+S_{22}^{-1}(\hat d_i,\hat\alpha,\hat\beta) \cdot \tilde y_i (c_{10}) ,\\
	\tilde y_{\text{MME},i }(c_{01})&= Q_{11}^{-1}(\hat d_i,\hat\alpha,\hat\beta) \cdot \tilde y_i (c_{01})+Q_{12}^{-1}(\hat d_i,\hat\alpha,\hat\beta) \cdot \tilde y_i (c_{00}), \\
	\tilde y_{\text{MME},i }(c_{00})&= Q_{21}^{-1}(\hat d_i,\hat\alpha,\hat\beta) \cdot \tilde y_i (c_{01})+Q_{22}^{-1}(\hat d_i,\hat\alpha,\hat\beta) \cdot \tilde y_i (c_{00})  ,
	\end{align*}
	where $\hat\alpha$, $\hat\beta$ are method-of-moments estimators obtained by Algorithm 1, $\hat d_i=\frac{\tilde d_i-(N_v-1)\hat\alpha}{1-\hat\alpha-\hat\beta}$ and  
	\begin{align*}
	S_{11}^{-1}(\hat d_i,\hat\alpha,\hat\beta)&= \frac{ 1}{p\big[1- [1-(1-\hat\beta)p]^{\hat d_i}  \big] }, \\
	S_{12}^{-1}(\hat d_i,\hat\alpha,\hat\beta)&= -\frac{ 1- (1-\hat\alpha p)^{N_v-1-\hat d_i}   }{p(1-\hat\alpha p)^{N_v-1-\hat d_i}\big[ 1- [1-(1-\hat\beta)p]^{\hat d_i} \big]  }, \\
	S_{21}^{-1}(\hat d_i,\hat\alpha,\hat\beta)&= -\frac{ [1-(1-\hat\beta)p]^{\hat d_i}-(1- p)^{\hat d_i} }{p(1- p)^{\hat d_i}\big[ 1- [1-(1-\hat\beta)p]^{\hat d_i} \big]  }, \\
	S_{22}^{-1}(\hat d_i,\hat\alpha,\hat\beta)&= \frac{ 1- (1- p)^{\hat d_i}-(1-\hat\alpha p)^{N_v-1-\hat d_i}\big[(1-(1-\hat\beta)p)^{\hat d_i}-(1- p)^{\hat d_i}\big] }{p(1- p)^{\hat d_i}(1-\hat\alpha p)^{N_v-1-\hat d_i}  \big[ 1- [1-(1-\hat\beta)p]^{\hat d_i} \big]  } ,\\
	\bm Q^{-1}(\hat d_i,\hat\alpha,\hat\beta)&= \frac{p}{1-p}\cdot \bm S^{-1}(\hat d_i,\hat\alpha,\hat\beta).
	\end{align*} 
	
	Noting that $\hat\alpha=\alpha+ \mathcal O_p(N_v^{-3/2}p^{-1/2})$ and $\hat\beta=\beta+\mathcal O_p(N_v^{-1/2}p^{1/2})$, we have $\tilde {\bar y}_\text{MME}(c_k) \text{ (unknown error rates)} -\tilde {\bar y}_\text{MME}(c_k) \text{ (known error rates)}=\mathcal O_p(N_v^{-1/2}p^{1/2})$, and
	$1-\mathbb P\left(\left|\tilde {\bar y}_\text{MME}(c_k) \text{ (unknown error rates)} -\tilde {\bar y}_\text{MME}(c_k) \text{ (known error rates)}\right| =\mathcal O (N_v^{-1/2}p^{1/2} g   )\right) 
	=   \mathcal O (e^{-g^2/2}), g=\omega(1)$.
	
	Thus, we obtain $\text{Bias}[\tilde {\bar y}_\text{MME}(c_k)] \text{ (unknown error rates)} -\text{Bias}[\tilde {\bar y} _\text{MME}(c_k)] \text{ (known error rates)}=\mathcal O(\max\{N_v^{-1/2}p^{1/2} g,e^{-g^2/2}\}),\text{ where } g=\omega(1)$. Let $g=\sqrt{2\bar d}$, we have $\text{Bias}[\tilde {\bar y}_\text{MME}(c_k)] \text{ (unknown error rates)} -\text{Bias}[\tilde {\bar y} _\text{MME}(c_k)] \text{ (known error rates)}=\mathcal O(\max\{N_v^{-1/2},e^{-\bar d}\})=o(1)$.
	
	Similarly, we can show  $\text{Var}[\tilde {\bar y}_\text{MME}(c_k)] \text{ (unknown error rates)} -\text{Var}[\tilde {\bar y} _\text{MME}(c_k)] \text{ (known error rates)}=o(1)$.
	
	(i) Unbiasedness.
	
	It suffices to show $\tilde {\bar y}_\text{MME}(c_k)$ is asymptotically unbiased when $\alpha$ and $\beta$ are known. By the definition of degree distribution and Chernoff bound,  we have $1-\mathbb P\left(\hat d_i=\Theta(1/p)\cap d_i=\Theta(1/p )\right)=o(e^{-\bar d} )$. Since $\mathbb E[ \tilde y_{A\&S,i}|\tilde {\bm A }]$ is bounded, it suffices to show $\mathbb E[\tilde y_{\text{MME},i }(c_{k})\cdot I_{\{\hat d_i=\Theta(1/p) \cap d_i=\Theta(1/p ) \}}]= y_i(c_k)+o(1)$.

	The conditional expectations of $\tilde y_{\text{MME},i }(c_k)$ are 
	\begin{align*}\label{eqc.2}
	\mathbb{E}\big[\tilde y_{\text{MME},i }(c_{11})|\tilde {\bm A}\big]=&\ \Big[ S_{11}^{-1}(\hat d_i,\alpha,\beta) \cdot \check p_i^e ( c_{11},c_{11} )+ S_{12}^{-1}(\hat d_i,\alpha,\beta) \cdot \check p_i^e ( c_{10},c_{11} )\Big] y_i(c_{11}) \\
	&+ \Big[ S_{11}^{-1}(\hat d_i,\alpha,\beta) \cdot \check p_i^e ( c_{11},c_{10} )+ S_{12}^{-1}(\hat d_i,\alpha,\beta) \cdot \check p_i^e ( c_{10},c_{10} )\Big] y_i(c_{10}) ,\\
	\mathbb{E}\big[\tilde y_{\text{MME},i }(c_{10})|\tilde {\bm A}\big]=&\ \Big[S_{21}^{-1}(\hat d_i,\alpha,\beta) \cdot \check p_i^e ( c_{11},c_{11} )+ S_{22}^{-1}(\hat d_i,\alpha,\beta) \cdot \check p_i^e ( c_{10},c_{11} )\Big] y_i(c_{11}) \\
	&+ \Big[ S_{21}^{-1}(\hat d_i,\alpha,\beta) \cdot \check p_i^e ( c_{11},c_{10} )+ S_{22}^{-1}(\hat d_i,\alpha,\beta) \cdot \check p_i^e ( c_{10},c_{10} )\Big] y_i(c_{10}), \\
	\mathbb{E}\big[\tilde y_{\text{MME},i }(c_{01})|\tilde {\bm A}\big]=&\ \Big[ Q_{11}^{-1}(\hat d_i,\alpha,\beta) \cdot \check p_i^e ( c_{01},c_{01} )+ Q_{12}^{-1}(\hat d_i,\alpha,\beta) \cdot \check p_i^e ( c_{00},c_{01} )\Big] y_i(c_{01}) \\
	&+ \Big[ Q_{11}^{-1}(\hat d_i,\alpha,\beta) \cdot \check p_i^e ( c_{01},c_{00} )+ Q_{12}^{-1}(\hat d_i,\alpha,\beta) \cdot \check p_i^e ( c_{00},c_{00} )\Big] y_i(c_{00}) ,\\
	\mathbb{E}\big[\tilde y_{\text{MME},i }(c_{00})|\tilde {\bm A}\big]=&\ \Big[ Q_{21}^{-1}(\hat d_i,\alpha,\beta) \cdot \check p_i^e ( c_{01},c_{01} )+ Q_{22}^{-1}(\hat d_i,\alpha,\beta) \cdot \check p_i^e ( c_{00},c_{01} )\Big] y_i(c_{01}) \\
	&+ \Big[ Q_{21}^{-1}(\hat d_i,\alpha,\beta) \cdot \check p_i^e ( c_{01},c_{00} )+ Q_{22}^{-1}(\hat d_i,\alpha,\beta) \cdot \check p_i^e ( c_{00},c_{00} )\Big] y_i(c_{00}),\numberthis
	\end{align*}
	where 
	\begin{align*}
	\check p_i^e(c_{11},c_{11})&=p\big[1-(1-p)^{d_i}-(1-p)^{\tilde d_i}+(1-p)^{d_i+\tilde d_i-\check d_i}\big] ,\\
	\check p_i^e(c_{10},c_{10})&=p(1-p)^{d_i+\tilde d_i-\check d_i},\\
	\check p_i^e(c_{11},c_{10})&=p(1-p)^{d_i}\big[1-(1-p)^{\tilde d_i-\check d_i}\big],\\
	\check p_i^e(c_{10},c_{11})&=p(1-p)^{\tilde d_i}\big[1-(1-p)^{d_i-\check d_i}\big],\\
	\check p_i^e(c_{01},c_{01})&=\frac{1-p}{p}\ \check p_i^e(c_{11},c_{11}),\\
	\check p_i^e(c_{00},c_{00})&=\frac{1-p}{p}\ \check p_i^e(c_{10},c_{10}),\\
	\check p_i^e(c_{01},c_{00})&=\frac{1-p}{p}\ \check p_i^e(c_{11},c_{10}),\\
	\check p_i^e(c_{00},c_{01})&=\frac{1-p}{p}\ \check p_i^e(c_{10},c_{11}).
	\end{align*}
	
	Next, we compute the expectation of $ \mathbb{E}\big[\tilde y_{\text{MME},i }(c_{k})\cdot I_{\{\hat d_i=\Theta(1/p)\cap d_i=\Theta(1/p)\}}|\tilde {\bm A}\big]$. For the exposure level $c_{11}$, notice that
	\begin{align*}
	&\ S_{11}^{-1}(\hat d_i,\alpha,\beta) \cdot \check p_i^e ( c_{11},c_{11} )+ S_{12}^{-1}(\hat d_i,\alpha,\beta) \cdot \check p_i^e ( c_{10},c_{11} )\\
	=&\  \frac{ [1-(1-p)^{d_i}]    }{ 1- [1-(1-\beta)p]^{\hat d_i}  }- \frac{ (1-p)^{\tilde d_i}\big[1-(1-p)^{d_i-\check d_i}\big]}{(1-\alpha p)^{N_v-1-\hat d_i}\big[ 1- [1-(1-\beta)p]^{\hat d_i} \big]  }.
	\end{align*}
	By Lemma \ref{ll1}, we have 
	\begin{align*}
	&\ \mathbb E\left[  \left\{\ S_{11}^{-1}(\hat d_i,\alpha,\beta) \cdot \check p_i^e ( c_{11},c_{11} )+ S_{12}^{-1}(\hat d_i,\alpha,\beta) \cdot \check p_i^e ( c_{10},c_{11} )\right\} \cdot I_{\{\hat d_i=\Theta(1/p)\cap d_i=\Theta(1/p) \}}\right]\\
	=&\ \frac{ [1-(1-p)^{d_i}]    }{ \mathbb E\left[1- [1-(1-\beta)p]^{\hat d_i} \right] } - \frac{\mathbb E\left[ (1-p)^{\tilde d_i}\big[1-(1-p)^{d_i-\check d_i}\big]\right]}{ \mathbb E\left[(1-\alpha p)^{N_v-1-\hat d_i}\big[ 1- [1-(1-\beta)p]^{\hat d_i} \big] \right ] }\\
	&\ +\mathcal O\left( p^{1/3} \right),
	\end{align*} 
	where 
	\begin{align*}
	\mathbb E\left[ (1-p)^{\tilde d_i}\big[1-(1-p)^{d_i-\check d_i}\big]\right] =& \ (1-\alpha p)^{N_v-1-d_i} \cdot \left[1-(1-\beta)p\right]^{ d_i}\cdot \left[1-(1-\beta p)^{d_i} \right],\\
	\mathbb E\left[1- [1-(1-\beta)p]^{\hat d_i} \right] =& \ 1- [1-(1-\beta)p]^{-\frac{\alpha(N_v-1) }{1-\alpha-\beta}}\\
	& \ \cdot \bigg \{\left[ (1-\beta)(1-(1-\beta)p)^{\frac{1}{1-\alpha-\beta}}+\beta \right]^{d_i}\\
	&\  \cdot \left[ \alpha(1-(1-\beta)p)^{\frac{1}{1-\alpha-\beta}}+(1-\alpha)\right]^{N_v-1-d_i} \bigg\},    \\
	\mathbb E\left[(1-\alpha p)^{N_v-1-\hat d_i}\right]=&\ (1-\alpha p)^{N_v-1+\frac{\alpha(N_v-1)}{1-\alpha-\beta}} \cdot  \left[ (1-\beta)(1-\alpha p)^{-\frac{1}{1-\alpha-\beta}}+\beta \right]^{d_i}\\
	&\  \cdot \left[ \alpha(1-\alpha p )^{-\frac{1}{1-\alpha-\beta}}+ (1-\alpha) \right]^{N_v-1-d_i}  ,\\
	\mathbb E\left[ \left( \frac{1-(1-\beta)p}{1-\alpha p} \right)^{\hat d_i}\right]=&\  \left( \frac{1-(1-\beta)p}{1-\alpha p} \right)^{ -\frac{\alpha(N_v-1) }{1-\alpha-\beta} }\cdot  \left[ (1-\beta)\left( \frac{1-(1-\beta)p}{1-\alpha p} \right)^{\frac{1}{1-\alpha-\beta}}+\beta \right]^{d_i}\\
	&\  \cdot \left[ \alpha\left( \frac{1-(1-\beta)p}{1-\alpha p} \right)^{ \frac{1}{1-\alpha-\beta}}+ (1-\alpha) \right]^{N_v-1-d_i} .
	\end{align*}
	Then, direct computation yields
	\begin{align*}
	\mathbb E\left[  \left\{\ S_{11}^{-1}(\hat d_i,\alpha,\beta) \cdot \check p_i^e ( c_{11},c_{11} )+ S_{12}^{-1}(\hat d_i,\alpha,\beta) \cdot \check p_i^e ( c_{10},c_{11} )\right\} \cdot I_{\{\hat d_i=\Theta(1/p)\cap d_i=\Theta(1/p)\}}\right]=1+\mathcal O(p^{1/3} ).    
	\end{align*}
	Similarly, we have 
	\begin{align*}
	\mathbb E\left[\left	\{S_{11}^{-1}(\hat d_i,\alpha,\beta) \cdot \check p_i^e ( c_{11},c_{10} )+ S_{12}^{-1}(\hat d_i,\alpha,\beta) \cdot \check p_i^e ( c_{10},c_{10} )\right\}  \cdot I_{\{\hat d_i=\Theta(1/p)\cap d_i=\Theta(1/p)\}}  \right] = \mathcal O(p^{1/3} ).
	\end{align*}
	Thus, we obtain $\mathbb E[\tilde y_{\text{MME},i }(c_{11}) \cdot I_{\{\hat d_i=\Theta(1/p)\cap d_i=\Theta(1/p)\}}  ] =y_i(c_{11}) +\mathcal O(p^{1/3} )$. Analogously, we can show $\mathbb{E}\big[\tilde y_{\text{MME},i }(c_{k})\cdot I_{\{\hat d_i=\Theta(1/p)\cap d_i=\Theta(1/p)\}}]=y_i(c_k)+\mathcal O(p^{1/3} )$ for other exposure levels. Therefore, we have $\text{Bias}(\tilde{\bar y}_{\text{MME} }(c_{k}))=\mathcal O(p^{1/3} )$.
	
	(ii) Consistency.
	
	Since $	\tilde {\overline{y}}_\text{MME}(c_k)$ is an asymptotically unbiased estimator of $\overline{y}(c_k)$ and $\text{Var}(\tilde {\overline{y}}_\text{MME}(c_k)) \text{ (unknown error rates) }- \text{Var}(\tilde {\overline{y}}_\text{MME}(c_k))\text{ (known error rates) }=o(1)$, it suffices to show $\text{Var}(\tilde {\overline{y}}_\text{MME}(c_k)) \text{ (known error rates) }=o(1)$.  Next, we compute $\text{Var}(\tilde {\overline{y}}_\text{MME}(c_k))$ when $\alpha$ and $\beta$ are known, i.e., $\hat d_i=\frac{\tilde d_i-(N_v-1) \alpha}{1- \alpha- \beta}.$
	
	By Cauchy-Schwarz inequality and the inequality $2uv\leq u^2+v^2$, we have
	\begin{align*}\text{Var}[\tilde {\overline{ {y}}}_\text{MME}(c_k)]\leq&\ 2\Bigg\{ \text{Var}\Big[ \frac{1}{N_v}\sum_{i=1}^{N_v}  \tilde {{ {y}}}_{ \text{MME},i}(c_k)\cdot I_{\{\hat d_i=\Theta(1/p)\}}\Big]\\
	&+\text{Var}\Big[ \frac{1}{N_v}\sum_{i=1}^{N_v}\tilde {{ {y}}}_{ \text{A\&S},i}(c_k)\cdot I_{\{\hat d_i=\omega(1/p) \bigcap c_k\in\{c_{11},c_{01} \} \}} \Big] \\
	&+\text{Var}\Big[ \frac{1}{N_v}\sum_{i=1}^{N_v}\tilde {{ {y}}}_{ \text{A\&S},i}(c_k)\cdot I_{\{\hat d_i=o(1/p) \bigcap c_k\in\{c_{10},c_{00} \} \}} \Big]\Bigg\} . 
	\end{align*}
	Following the proof of Theorem 2, we can show $\text{Var}\Big[ \frac{1}{N_v}\sum_{i=1}^{N_v}\tilde {{ {y}}}_{ \text{A\&S},i}(c_k)\cdot I_{\{\hat d_i=\omega(1/p)  \}} \Big]=o(1)$ and $\text{Var}\Big[ \frac{1}{N_v}\sum_{i=1}^{N_v}\tilde {{ {y}}}_{ \text{A\&S},i}(c_k)\cdot I_{\{\hat d_i=o(1/p) \}} \Big]=o(1)$. Next, we prove $\text{Var}\Big[ \frac{1}{N_v}\sum_{i=1}^{N_v}  \tilde {{ {y}}}_{ \text{MME},i}(c_k)\cdot I_{\{\hat d_i=\Theta(1/p)\}}\Big]=o(1)$.
	
	For the exposure level $c_{00}$, we have
	\begin{align*}
	\tilde y_{\text{MME},i }(c_{00})= Q_{21}^{-1}(\hat d_i,\alpha,\beta)\tilde{p}_{i}(c_{01}) \cdot \tilde y_\text{A\&S,i} (c_{01})+Q_{22}^{-1}(\hat d_i,\alpha,\beta) \tilde{p}_{i}(c_{00})\cdot \tilde y_\text{A\&S,i}(c_{00}),
	\end{align*}
	where $Q_{21}^{-1}(\hat d_i,\alpha,\beta)\tilde{p}_{i}(c_{01})$ and $Q_{22}^{-1}(\hat d_i,\alpha,\beta) \tilde{p}_{i}(c_{00})$ are bounded when $\hat d_i=\Theta(1/p)$. Again, by Cauchy-Schwarz inequality and the inequality $2uv\leq u^2+v^2$, we obtain
	\begin{align*}\label{eqc.3}
	&\ \text{Var}\Big[ \frac{1}{N_v}\sum_{i=1}^{N_v}  \tilde {{ {y}}}_{ \text{MME},i}(c_{00} )\cdot I_{\{\hat d_i=\Theta(1/p)\}}\Big]\\
	\leq &\ 2 \Bigg\{\text{Var}\Big[ \frac{1}{N_v}\sum_{i=1}^{N_v} Q_{21}^{-1}(\hat d_i,\alpha,\beta)\tilde{p}_{i}(c_{01})  \tilde y_\text{A\&S,i} (c_{01})\cdot I_{\{\hat d_i=\Theta(1/p)\}}\Big] \\
	& + \text{Var}\Big[ \frac{1}{N_v}\sum_{i=1}^{N_v} Q_{22}^{-1}(\hat d_i,\alpha,\beta) \tilde{p}_{i}(c_{00})  \tilde y_\text{A\&S,i}(c_{00})\cdot I_{\{\hat d_i=\Theta(1/p)\}}\Big]   \Bigg\}.\numberthis
	\end{align*}
	Thus, it suffices to show the two variances in (\ref{eqc.3}) go to zero as $N_v\rightarrow\infty$. Here, we show $\text{Var}\Big[ \frac{1}{N_v}\sum_{i=1}^{N_v} Q_{22}^{-1}(\hat d_i,\alpha,\beta) \tilde{p}_{i}(c_{00})  \tilde y_\text{A\&S,i}(c_{00})\cdot I_{\{\hat d_i=\Theta(1/p)\}}\Big] =o(1)$. The other one can be proved similarly.

	By the law of total variance, we have
	\begin{align*}
	&\ \text{Var}\Big[ \frac{1}{N_v}\sum_{i=1}^{N_v} Q_{22}^{-1}(\hat d_i,\alpha,\beta) \tilde{p}_{i}(c_{00})  \tilde y_\text{A\&S,i}(c_{00})\cdot I_{\{\hat d_i=\Theta(1/p)\}}\Big]\\
	=&\ \text{Var}\Big[ \mathbb E\Big(\frac{1}{N_v}\sum_{i=1}^{N_v}  Q_{22}^{-1}(\hat d_i,\alpha,\beta) \tilde{p}_{i}(c_{00})  \tilde y_\text{A\&S,i}(c_{00})\cdot I_{\{\hat d_i=\Theta(1/p)\}}|\tilde{\bm A} \Big)\Big]\\
	&\ +\mathbb E \Big[ \text{Var}\Big(\frac{1}{N_v}\sum_{i=1}^{N_v} Q_{22}^{-1}(\hat d_i,\alpha,\beta) \tilde{p}_{i}(c_{00})  \tilde y_\text{A\&S,i}(c_{00})\cdot I_{\{\hat d_i=\Theta(1/p)\}}|\tilde{\bm A} \Big)\Big].
	\end{align*}
	Since $Q_{22}^{-1}(\hat d_i,\alpha,\beta) \tilde{p}_{i}(c_{00})\cdot I_{\{\hat d_i=\Theta(1/p)\}}=O(1)$, we obtain
	\begin{align*}
	& \text{Var}\Big[ \mathbb E\Big(\frac{1}{N_v}\sum_{i=1}^{N_v}  Q_{22}^{-1}(\hat d_i,\alpha,\beta) \tilde{p}_{i}(c_{00})  \tilde y_\text{A\&S,i}(c_{00})\cdot I_{\{\hat d_i=\Theta(1/p)\}}|\tilde{\bm A} \Big)\Big] \\
	= &\  \mathcal O\Bigg(  \text{Var}\Big[ \mathbb E\Big(\frac{1}{N_v}\sum_{i=1}^{N_v}   \tilde y_\text{A\&S,i}(c_{00})|\tilde{\bm A} \Big)\Big]\Bigg).
	\end{align*}
	Notice that $\text{Var}\Big[ \mathbb E\Big(\frac{1}{N_v}\sum_{i=1}^{N_v}   \tilde y_\text{A\&S,i}(c_{00})|\tilde{\bm A} \Big)\Big]=o(1)$ (shown in the proof of Theorem 2), we obtain 
	\begin{align*}
	\text{Var}\Big[ \mathbb E\Big(\frac{1}{N_v}\sum_{i=1}^{N_v}  Q_{22}^{-1}(\hat d_i,\alpha,\beta) \tilde{p}_{i}(c_{00})  \tilde y_\text{A\&S,i}(c_{00})\cdot I_{\{\hat d_i=\Theta(1/p)\}}|\tilde{\bm A} \Big)\Big]=o(1).
	\end{align*}
	Following the proof of Theorem 2, we can show
	\begin{align*}
	&\ \text{Var}\Big(\frac{1}{N_v}\sum_{i=1}^{N_v} Q_{22}^{-1}(\hat d_i,\alpha,\beta) \tilde{p}_{i}(c_{00})  \tilde y_\text{A\&S,i}(c_{00})\cdot I_{\{\hat d_i=\Theta(1/p)\}}|\tilde{\bm A} \Big)\\
	\leq&\ \frac{C_1}{N_v^2}    \sum_{i=1}^{N_v} \frac{1}{\tilde p_i^e(c_{00})}+\frac{C_2}{N_v^2} \sum_{i=1}^{N_v} \sum_{j\neq i }\frac{\tilde g_{ij}\tilde p_{ij}^e(c_{00}) }{\tilde p_i^e(c_{00})\tilde p_j^e(c_{00})},
	\end{align*}
	where $C_1$ and $C_2$ are positive constants. By Proposition 4 and (\ref{eqb.1}), we obtain 
	\begin{align*}
	\mathbb E \Big[ \text{Var}\Big(\frac{1}{N_v}\sum_{i=1}^{N_v} Q_{22}^{-1}(\hat d_i,\alpha,\beta) \tilde{p}_{i}(c_{00})  \tilde y_\text{A\&S,i}(c_{00})\cdot I_{\{\hat d_i=\Theta(1/p)\}}|\tilde{\bm A} \Big)\Big]=o(1).
	\end{align*}
	These complete the proof.
	
	\subsection{Proof of Theorem 5}
	
	The proof of Theorem 5 is same as that of Theorem 4.
	
	\subsection{Proof of Theorem 6}
	
	We first introduce a useful theorem. \cite{baldi1989normal} proves the following:
	
	Let $\{Z_i,i\in V\}$ be random variables having a dependency graph $G^*=(V,E^*)$. For $i\in V$, let $L_i^{(k)}$ denote the number of connected subsets of $V$ of cardinality at most $k$ which contain $i$. Let $W=\sum_{i\in V} Z_i$ and $\sigma^2=\text{Var} (W)<\infty.$ Set 
	\begin{align} \label{eq:a}
	\frac{1}{\sigma^k} \sum_{i\in V}  \mathbb E \left(L_i^{(k)}|Z_i-\mathbb E Z_i|^k\right)= A_k <\infty,
	\ k=3,4.
	\end{align}
	Then for all real $w$, 
	\begin{align}\label{eq:clt1}
	\Big|\mathbb P\Big( \frac{W-\mathbb E W}{\sigma}\leq w \Big)-\Phi(w)\Big| &\leq  c \ (\sqrt{A_3} + \sqrt{A_4} ),
	\end{align}
	for some constant $0\leq c <8$ which does not depend on $\{Z_i\}$.

	Recall that $G^*$ is said to be a dependency graph if  for any pair of disjoint sets $V_1, V_2$ in $V$ such that no edge in $E^*$ has one endpoint in $V_1$ and the other in $V_2$, the sets of random variables $\{X_i,i\in V_1 \}$ and $\{X_i,i\in V_2 \}$ are independent. Next, we show the following lemma by the above theorem.
	
	\begin{lemma}\label{l3}
		Let $\{Z_i,i\in V\}$ denote random variables having a dependency graph $G^*=(V,E^*)$. For $i\in V$, let $L_i^{(k)}$ denote the number of connected subsets of $V$ of cardinality at most $k$ which contain $i$. Let $W=\sum_{i\in V} Z_i$ and $\sigma^2=\text{Var} (W)<\infty$. Suppose that $\mathbb E |Z _i - \mathbb E Z_i|^k\leq B_k$, $k=6,8$. Then for all real $w$, 
		\begin{align*}
		\Big|\mathbb P\Big( \frac{W-\mathbb E W}{\sigma}\leq w \Big)-\Phi(w)\Big|  = &\ \mathcal O \Bigg(\sqrt{\frac{B_6^{1/2} }{\sigma^3} \left(\sum_{i\in V}  \mathbb E [(L_i^{(3)})^2 ]\right)^{1/2} }  \\
		&\ + \sqrt{\frac{B_8^{1/2} }{\sigma^4} \left(\sum_{i\in V}  \mathbb E [(L_i^{(4)})^2 ]\right)^{1/2} }  \Bigg) .
		\end{align*}
	\end{lemma}
	
	\begin{proof}
		By Cauchy–Schwarz inequality, we have
		\begin{align*}
		\mathbb E \left(L_i^{(k)}|Z_i-\mathbb E Z_i|^k\right) \leq \sqrt{ \mathbb E[(L_i^{(k)})^2 ] \cdot \mathbb E|Z_i-\mathbb E Z_i|^{2k} }
		\end{align*}
		Then by (\ref{eq:clt1}), we can show the lemma.
	\end{proof}

	Now consider our method-of-moments estimators. Note that $\frac{\tilde{\bar y}_\text{MME} (c_k)\text{(known error rates)}-\tilde{\bar y}_\text{MME} (c_k)\text{(unknown error rates)}  }{\sqrt{\text{Var}( \tilde{\bar y}_\text{MME} (c_k) )  } \text{(unknown error rates)}}\xrightarrow{P} 0$. By Lemma \ref{l5}, we obtain that $\frac{\tilde{\bar y}_\text{MME} (c_k)-\mathbb E( \tilde{\bar y}_\text{MME} (c_k)  )}{\sqrt{\text{Var}( \tilde{\bar y}_\text{MME} (c_k) )  } } \text{ (known error rates)}$ and $\frac{\tilde{\bar y}_\text{MME} (c_k)-\mathbb E( \tilde{\bar y}_\text{MME} (c_k)  )}{\sqrt{\text{Var}( \tilde{\bar y}_\text{MME} (c_k) )  } } \text{ (unknown error rates)}$ have the same asymptotic distribution, and $\text{Var}( \tilde{\bar y}_\text{MME} (c_k) )\text{(known error rates)}\sim \text{Var}( \tilde{\bar y}_\text{MME} (c_k) )\text{(unknown error rates)}$.  Thus it suffices to show the asymptotic normality of $\frac{\tilde{\bar y}_\text{MME} (c_k)-\mathbb E( \tilde{\bar y}_\text{MME} (c_k)  )}{\sqrt{\text{Var}( \tilde{\bar y}_\text{MME} (c_k) )  } } $ for known error rates. 
	
	Next we show $\frac{\tilde{\bar y}_\text{MME} (c_k)-\mathbb E( \tilde{\bar y}_\text{MME} (c_k)  )}{\sqrt{\text{Var}( \tilde{\bar y}_\text{MME} (c_k) )  } } \xrightarrow{d} N(0,1)$ for known error rates. Let $G^*=(V,E^*)$ be the union of the true graph $G=(V,E)$ and the observed graph $G^\text{obs}=(V,E^\text{obs})$. For $i\in V$, let $L_i^{(k)}$ denote the number of connected subsets of $V$ of cardinality at most $k$ which contain $i$. Let $Z_i(c_k)$ be the estimator of $y_i(c_k)/N_v$, i.e., 
	\begin{align}\label{eq:defZ}
	\begin{split}
	Z_i(c_k) =  &\  \frac{1}{N_v} \Big\{  \tilde {{ {y}}}_{ \text{MME},i}(c_k)\cdot I_{\{\hat d_i=\Theta(1/p)\}}+\tilde {{ {y}}}_{ \text{A\&S},i}(c_k)\cdot I_{\{\hat d_i= \omega(1/p) \bigcap c_k\in\{c_{11},c_{01} \} \}} \\
	&+\tilde {{ {y}}}_{ \text{A\&S},i}(c_k)\cdot I_{\{\hat d_i=o(1/p) \bigcap c_k\in\{c_{10},c_{00} \}  \}} \Big\} .
	\end{split}
	\end{align}
	Then, $G^*$ is a dependence graph of the set of random variables $\{ Z_i (c_k), i\in V\}$. Define $W(c_k)=\sum_{i\in V} Z_i(c_k)$, $\sigma^2(c_k)=\text{Var} (W(c_k))$. By Theorem 5, we have $\sigma^2(c_k)<\infty$. Let $C_{\mathcal V_1}$ be the two-stars counts in $G$, $C_{\mathcal V_2}$ counts of 3 connected edges in $G$ passing through 4 different nodes. By the definition of $G^*$, we have $(\sum_{i\in V}  \mathbb E [(L_i^{(3)})^2 ])^{1/2}=\mathcal O(C_{\mathcal V_1})$ and $(\sum_{i\in V}  \mathbb E [(L_i^{(4)})^2 ])^{1/2}=\mathcal O(C_{\mathcal V_1}+C_{\mathcal V_2})$. Let $B_6(c_k)$ be an upper bound of $\mathbb E|Z_i(c_k)-\mathbb E Z_i(c_k)|^6$, $B_8(c_k)$ be an upper bound of $\mathbb E|Z_i(c_k)-\mathbb E Z_i(c_k)|^8$, $i\in V$. Then, by Lemma \ref{l3}, we have 
	\begin{align*}
	\Big|\mathbb P\Big( \frac{W(c_k)-\mathbb E W(c_k)}{\sigma(c_k)}\leq w \Big)-\Phi(w)\Big| &=  \mathcal O   \Big(\sqrt{\frac{B_6^{1/2}(c_k)}{\sigma^3(c_k)} C_{\mathcal V_1}} + \sqrt{ \frac{B_8^{1/2}(c_k)}{\sigma^4(c_k)} (C_{\mathcal V_1}+C_{\mathcal V_2}) } \Big) .
	\end{align*}
	Therefore, it suffices to show that 
	\begin{align*}
	\frac{B_6^{1/2}(c_k)}{\sigma^3(c_k)} C_{\mathcal V_1} =o(1), \frac{B_8^{1/2}(c_k)}{\sigma^4(c_k)} C_{\mathcal V_1}  =o(1) , \text{ and }    \frac{B_8^{1/2}(c_k)}{\sigma^4(c_k)}  C_{\mathcal V_2} =o(1). 
	\end{align*}
	
	By the definition of $Z_i(c_k)$, we can easily obtain that 
	\begin{align*}
	B_6 (c_{k} )  =\mathcal O\Big(   \frac{1}{p^6N_v^6 }   \Big), \text{ and }B_8 (c_{k} )=\mathcal O\Big(   \frac{1}{p^8N_v^8  }   \Big), \text{ where } c_k \in \{ c_{11},c_{10} \},\\
	B_6 (c_{k} )  =\mathcal O\Big(   \frac{1}{N_v^6 }   \Big), \text{ and }B_8 (c_{k} )=\mathcal O\Big(   \frac{1}{N_v^8  }   \Big), \text{ where } c_k \in \{ c_{01},c_{00} \}.
	\end{align*}
	Note that $C_{\mathcal V_1}=\Theta(N_v(\bar d)^2)$. Direct calculation yields to
	\begin{align*}
	\frac{B_6^{1/2}(c_k)}{\sigma^3(c_k)} C_{\mathcal V_1} =o(1), \frac{B_8^{1/2}(c_k)}{\sigma^4(c_k)} C_{\mathcal V_1}  =o(1) , \text{ and }    \frac{B_8^{1/2}(c_k)}{\sigma^4(c_k)}  C_{\mathcal V_2} =o(1).
	\end{align*}
	Thus, we show 
	$$\Big|\mathbb P\Big( \frac{W(c_k)-\mathbb E W(c_k)}{\sigma(c_k)}\leq w \Big)-\Phi(w)\Big|  =  o(1).$$
	
	\subsection{Proof of Theorem 7}
	
	The proof of Theorem 7 is same as that of Theorem 6.

	\subsection{Proof of Corollary 2}
	
	Under Assumptions 1 -- 7, by Slutsky's theorem, we obtain that $\frac{\tilde{\bar y}_\text{MME} (c_k)-{\bar y} (c_k) }{\sqrt{\text{Var}( \tilde{\bar y}_\text{MME} (c_k) )  } }$ and $\frac{\tilde{\bar y}_\text{MME} (c_k)-\mathbb E( \tilde{\bar y}_\text{MME} (c_k)  )}{\sqrt{\text{Var}( \tilde{\bar y}_\text{MME} (c_k) )  } }$ have the same asymptotic distribution. Then, by Theorem 6, we have $\frac{\tilde{\bar y}_\text{MME} (c_k)-{\bar y} (c_k) }{\sqrt{\text{Var}( \tilde{\bar y}_\text{MME} (c_k) )  } }\xrightarrow{d}N(0,1)$.

	\section{Proofs of theorems in the generalized four-level exposure model}
	\subsection{Proof of Theorem 8}
	In the generalized four-level exposure model and Bernoulli random assignment of treatment with $p$, for each individual $i$, the four exposure probabilities are as follows.
	\begin{equation}
	\begin{aligned}
	p_i^e(c_{11'})&=p \sum_{x=m_i}^{d_i}\binom{d_i}{x}p^x(1-p)^{d_i-x}, \\
	p_i^e(c_{10'})&=p\sum_{x=0}^{(m_i-1)\wedge d_i}\binom{d_i}{x}p^x(1-p)^{d_i-x},\\
	p_i^e(c_{01'})&=(1-p) \sum_{x=m_i}^{d_i}\binom{d_i}{x}p^x(1-p)^{d_i-x}, \\
	p_i^e(c_{00'})&=(1-p)\sum_{x=0}^{(m_i-1)\wedge d_i}\binom{d_i}{x}p^x(1-p)^{d_i-x}.
	\end{aligned}
	\end{equation}

	For $m_i=1$, direct computations lead to the results in Theorem 8. Here, we show the case when $m_i\geq 2$.
	Note that, for all $1\leq x\leq d_i$, we have (\cite{das2016brief})
	\begin{align*}
	\Big(\frac{d_i}{x}\Big)^x\leq \binom{d_i}{x}\leq  \frac{d_i^x}{x!} .
	\end{align*}
	For the exposure level $c_{00'}$, we obtain
	\begin{align*}
	p_i^e(c_{00'})&\geq \frac{1-p}{(m_i-1)^{m_i-1}} \sum_{x=1}^{m_i-1} (d_ip)^x(1-p)^{d_i-x}+(1-p)^{d_i+1} \\
	&\geq C_1(1-p)^{d_i+1} \frac{(d_ip)^{m_i}(1-p)^{-(m_i-1)} -(1-p)}{p(d_i+1)-1}
	\end{align*}
	and
	\begin{align*}
	p_i^e(c_{00'})\leq  (1-p) \sum_{x=1}^{m_i-1} (d_ip)^x(1-p)^{d_i-x}+(1-p)^{d_i+1}  \leq C_2(1-p)^{d_i+1} \frac{(d_ip)^{m_i}(1-p)^{-(m_i-1)} -(1-p)}{p(d_i+1)-1},
	\end{align*}
	where $C_1$ and $C_2$ are positive constants.
	Thus, if $p=o(1)$, we have
	\begin{align*}
	p_i^e(c_{00'})& =\begin{cases}
	\Theta\Big(	\dfrac{ (d_ip)^{m_i-1}}  {e^{d_ip} } \Big), & d_i=\omega(1/p),\\
	\Theta\Big(1 \Big), & d_i=\Theta(1/p),\\
	\Theta\Big(1 \Big), & d_i=o(1/p).
	\end{cases} 
	\end{align*}
	
	For the exposure level $c_{01'}$, we have
	\begin{align*}
	p_i^e(c_{01'})\geq (1-p) \sum_{x=1}^{m_i-1} p^x(1-p)^{d_i-x}   \geq C_3(1-p)  \frac{p^{m_i}(1-p)^{d_i-(m_i-1)} -p^{d_i+1}  }{1-2p}
	\end{align*}
	and
	\begin{align*}
	p_i^e(c_{01'})\leq  \frac{1-p}{m_i!}  \sum_{x=m_i}^{d_i} (d_ip)^x(1-p)^{d_i-x}   \leq C_4(1-p)  \frac{(d_ip)^{d_i+1}-(d_ip)^{m_i} (1-p)^{d_i-(m_i-1)}  }{p(d_i+1)-1},
	\end{align*}
	where $C_3$ and $C_4$ are positive constants. Together with $p_i^e(c_{01'})=1-p-p_i^e(c_{00'})$, we obtain 
	\begin{align*}
	p_i^e(c_{01'})& =\begin{cases}
	\Theta\Big(1\Big), & d_i=\omega(1/p),\\
	\mathcal O\Big(1 \Big) \text{ and } \Omega(p^{m_i} ), & d_i=\Theta(1/p),\\
	\mathcal O\Big( (d_ip)^{m_i}   \Big)  \text{ and } \Omega(p^{m_i} ), & d_i=o(1/p).
	\end{cases} 
	\end{align*}
	Note that $p_i^e(c_{11'}) =p/(1-p)\cdot p_i^e(c_{01'})$ and $p_i^e(c_{10'}) =p/(1-p)\cdot p_i^e(c_{00'})$. The results for exposure levels $c_{11'}$ and $c_{10'}$ follow.
	
	\section{Proofs for Pareto degree distribution without a cutoff}
	
	In this section, we show that assume a four-level exposure model and Bernoulli random assignment of treatment with $p$. In the inhomogeneous graph where the asymptotic degree distribution is the Pareto distribution with shape $\zeta>1$, lower bound $d_L$, upper bound $N_v-1$ and mean $\bar d$, under Assumption 4, Condition 1 doesn't hold for levels $c_{10}$. 
	
	As $N_v\rightarrow\infty$, we have
	\begin{align*}\label{eq:E1}
	\mathbb E[1/p^e(c_{10})]& =\frac{1}{p}\cdot  \frac{\zeta d_L^\zeta }{1- \Big(\dfrac{d_L}{N_v-1}\Big )^\zeta } \cdot \int_{d_L}^{N_v-1} (1-p)^{-x}  x^{-(\zeta+1)}dx  \\
	& =\Theta \Bigg( N_v\cdot   \frac{d_L^\zeta }{p N_v} \cdot  \int_{d_L}^{N_v-1} (1-p)^{-x}  x^{-(\zeta+1)}dx \Bigg)  . \numberthis
	\end{align*}	
	Note that $\lim_{x\rightarrow\infty}(1-p)^{-x}  x^{-(\zeta+1)}=\lim_{x\rightarrow\infty} p^{\zeta+1} e^{px}$. By (\ref{eq:E1}), we have 
	\begin{align}\label{eq:E2}
	\mathbb E[1/p^e(c_{10})]=\Omega(N_v\cdot (d_Lp)^\zeta\cdot e^{pN_v} ).
	\end{align}
	Next, we compute the order of $d_L$. By the definition of expectation, we have
	\begin{align*}  
	\int_{d_L}^{N_v-1}x\cdot \frac{\zeta d_L^\zeta }{1- \Big(\dfrac{d_L}{N_v-1}\Big )^\zeta } x^{-(\zeta+1)}dx
	&=\begin{cases}
	\zeta d_L^\zeta \cdot  \dfrac{\log\Big(\dfrac{N_v-1}{d_L}\Big ) } {1- \Big(\dfrac{d_L}{N_v-1}\Big )^\zeta } ,   \text{ if }  \zeta=1\\
	\dfrac{\zeta d_L}{\zeta-1}\cdot  \dfrac{1- \Big(\dfrac{d_L}{N_v-1}\Big )^{\zeta-1} } {1- \Big(\dfrac{d_L}{N_v-1}\Big )^\zeta } ,   \text{ otherwise, }\\
	\end{cases}
	\end{align*}
	Therefore, as $N_v\rightarrow\infty$, we obtain $d_L=\Theta(\bar d)$ when $\zeta>1$. By (\ref{eq:E2}), we have $\mathbb E[1/p^e(c_{10})]=\Omega(N_v)$. Therefore, condition 1 doesn't hold.